\documentclass[11pt]{article}%
\usepackage{amsmath}
\usepackage{amsfonts}
\usepackage{amssymb}
\usepackage{graphicx}%
\providecommand{\U}[1]{\protect\rule{.1in}{.1in}}
\newtheorem{theorem}{Theorem}

\newtheorem{corollary}[theorem]{Corollary}

\newtheorem{example}[theorem]{Example}

\newtheorem{lemma}[theorem]{Lemma}

\newtheorem{proposition}[theorem]{Proposition}
\newtheorem{remark}[theorem]{Remark}

\newenvironment{proof}[1][Proof]{\noindent\textbf{#1.} }{\ \rule{0.5em}{0.5em}}

\begin{document}

\title{A refinement of the Robertson-Schr\"odinger uncertainty principle and a Hirschman-Shannon inequality for Wigner distributions}
\author{Nuno Costa Dias\textbf{\thanks{ncdias@meo.pt}}
\and Maurice A. de Gosson\textbf{\thanks{maurice.de.gosson@univie.ac.at}}
\and Jo\~{a}o Nuno Prata\textbf{\thanks{joao.prata@mail.telepac.pt }}}
\maketitle

\begin{abstract}
We propose a refinement of the Robertson-Schrodinger uncertainty principle
(RSUP) using Wigner distributions. This new principle is stronger than the RSUP.
In particular, and unlike the RSUP, which can be saturated by many phase space functions,
the refined RSUP can be saturated by pure Gaussian Wigner
functions only. Moreover, the new principle is technically as
simple as the standard RSUP. In addition, it makes a direct
connection with modern harmonic analysis, since it involves the
Wigner transform and its symplectic Fourier transform, which is
the radar ambiguity function.

As a by-product of the refined RSUP, we derive inequalities involving the entropy and the covariance matrix of Wigner distributions. These inequalities refine the Shanon and the Hirschman inequalities for the Wigner distribution of a mixed quantum state $\rho$. We prove sharp
estimates which critically depend on the purity of $\rho$ and which are
saturated in the Gaussian case.
\end{abstract}

\section{Introduction}

In quantum mechanics, the state of a system is represented by a positive trace class operator with unit trace - called a density matrix - acting on a separable Hilbert space $\mathcal{H}$. We denote the set of density matrices - the set of states - by $\mathcal{S}(\mathcal{H})$. Given some trace class operator $\widehat{\rho}$, it is in general very difficult to assess whether $ \widehat{\rho} \in \mathcal{S}(\mathcal{H})$. The main difficulty resides in the verification of the positivity condition:
\begin{equation}
(f | \widehat{\rho}  f )_{\mathcal{H}} \ge 0,
\label{eqIntroduction1}
\end{equation}
for all $f \in \mathcal{H}$. This is particularly difficult in infinite dimensional Hilbert spaces. In this work we shall be concerned with the case $\mathcal{H}=L^2 (\mathbb{R}^n)$.

A very useful representation of density matrices, which casts position and momentum variables on equal footing and is akin to a classical probability density, is the Wigner distribution \cite{Wigner}. It is obtained from $\widehat{\rho}$ by way of the Weyl transform \cite{Birk,Wong}:
\begin{equation}
\widehat{\rho} \mapsto W \rho (x,p) = \frac{1}{(2 \pi \hbar)^n} \int_{\mathbb{R}^n} \rho \left(x+ \frac{y}{2},x- \frac{y}{2} \right) e^{- \frac{i}{\hbar} p \cdot y} dy,
\label{eqIntroduction2}
\end{equation}
where $\rho ( \cdot, \cdot ) \in L^2 (\mathbb{R}^{2n})$ is the Hilbert-Schmidt kernel of $\widehat{\rho}$. Here $h=2 \pi \hbar$ is Planck's constant and $x,p$ denote the particle's position and momentum respectively. We shall write them collectively as $z=(x,p) \in \mathbb{R}^{2n}$, a point in the particle's phase space $\mathbb{R}^n \times (\mathbb{R}^n)^{\ast} \simeq \mathbb{R}^{2n}$.

The Wigner distribution is not a true probability density as it may be negative \cite{Gro,Hudson}. Rather, it defines a finite signed measure:
\begin{equation}
A \mapsto \mu_{\rho} (A):= \int_{A} W \rho (x,p) dx dp ,
\label{eqIntroduction3}
\end{equation}
for Borel sets $A \in \mathcal{B}(\mathbb{R}^{2n})$, and $ \mu_{\rho} (\mathbb{R}^{2n})=1$.

This means that the covariance matrix $\operatorname*{Cov}(W \rho)$ of $W \rho$ might {\it a priori} not be positive definite. However, it can be shown that it is \cite{Narcow}. In fact, it obeys an even stronger constraint called the Robertson-Schr\"odinger uncertainty principle (RSUP) which states that \cite{Birk,Narcow2,Narcow3,Narconnell}
\begin{equation}
\operatorname*{Cov}(W \rho) + \frac{i \hbar}{2} J \ge 0,
\label{eqIntroduction4}
\end{equation}
where $J$ is the standard symplectic matrix:
\begin{equation}
J= \left(
\begin{array}{c c}
0 & I\\
-I & 0
\end{array}
\right).
\label{eqIntroduction5}
\end{equation}
It can be shown that condition (\ref{eqIntroduction4})
is a necessary but not sufficient condition for a phase space
function to be a Wigner distribution \cite{PLA}.

Nevertheless it has many interesting features. For a Gaussian
measure $G$  it is both a necessary and sufficient condition for
$G$ to be a Wigner distribution \cite{Narcow2}. It is invariant
under linear symplectic transformations (unlike the more
frequently used Heisenberg uncertainty relation). It has a nice
geometric interpretation in terms of Poincar\'e invariants
\cite{Narcow}, and it is intimately related with symplectic
topology and Gromov's non-squeezing theorem
\cite{physreps,Gromov}. By a suitable linear symplectic
transformation, the RSUP makes it a simple task to determine
directions in phase space of minimal uncertainty \cite{Narcow}. In
particular, we say that the RSUP is {\it saturated} if we can find
$n$ two-dimensional symplectic planes, where the uncertainty is
minimal. More specifically, the RSUP (\ref{eqIntroduction4}) is
saturated, whenever all the Williamson invariants of
$\operatorname*{Cov}(W \rho)$ are minimal \cite{physreps,Narcow2}:
\begin{equation}
\lambda_{\sigma,1}(\operatorname*{Cov}(W \rho))=\lambda_{\sigma,2}(\operatorname*{Cov}(W \rho))= \cdots=\lambda_{\sigma, n}(\operatorname*{Cov}(W \rho))= \frac{\hbar}{2}.
\label{eqIntroduction12}
\end{equation}

Having said that, there is nothing about inequality
(\ref{eqIntroduction4}) which is particularly quantum mechanical,
with the exception of the presence of Planck's constant. In fact,
(\ref{eqIntroduction4}) is only a requirement about a minimal
scale related to $\hbar$. This condition is not sufficient
to ensure that the state is quantum mechanical (not even if
saturated). We shall give an example of a function in phase space
which saturates the RSUP, but which is manifestly not a Wigner
function. More emphatically, we will show that {\it any}
measurable function in phase space $F$ with a positive definite
covariance matrix $ \operatorname*{Cov}(F)>0$ satisfies
(\ref{eqIntroduction4}) after a suitable dilation $F(z) \mapsto
\lambda^{2n} F(\lambda z)$, while most of them remain non
quantum. This means that being a quantum state is not only a
question of scale but also of shape. This prompted us to
look for an alternative
uncertainty principle which goes beyond the RSUP.

In order to state our results precisely, let us fix some notation. In the sequel
$\mathcal{F}_{\sigma} (F)$ denotes the symplectic Fourier
transform of the function $F$. Roughly speaking, it can be
obtained from the ordinary Fourier transform $\mathcal{F} (F)$ by
a symplectic rotation and a dilation $(\mathcal{F}_{\sigma} F)
(z)= \frac{1}{(2 \pi \hbar)^n} (\mathcal{F} F) \left(\frac{Jz}{2
\pi \hbar} \right)$.

For a given measurable phase-space function $F$, satisfying
\begin{equation}
\int_{\mathbb{R}^{2n}} F(z) dz \ne 0,
\label{eqIntroduction7}
\end{equation}
we write
\begin{equation}
\widetilde{F} (z):= \frac{F(z)}{\int_{\mathbb{R}^{2n}} F(z) dz}~ .
\label{eqIntroduction7.1}
\end{equation}

Morevover, we denote by
\begin{equation}
<z>_F= \int_{\mathbb{R}^{2n}} z \widetilde{F}(z) dz
\label{eqIntroduction5}
\end{equation}
the expectation value of $z$ regarded as a column vector, and by
\begin{equation}
 \operatorname*{Cov}(F)= \int_{\mathbb{R}^{2n}} (z - <z>_F) (z - <z>_F)^T \widetilde{F}(z) dz
\label{eqIntroduction6}
\end{equation}
the covariance matrix. Notice that there is some abuse of language in this probabilistic terminology, as $F$ is not required to be non-negative.

The main result of this paper is Theorem \ref{TheoremERSUP2}, where we prove the following uncertainty principle, hereafter called {\it refined Robertson-Schr\"odinger uncertainty principle}:
\begin{equation}
\begin{array}{c}
 \operatorname*{Cov}(W \rho) + \frac{i \hbar}{2}J \ge  \\
 \\
 \ge \mathcal{P} \left[W \rho \right]  \left( \operatorname*{Cov}(|\widetilde{W \rho}|^2)+\frac{1}{4} \operatorname*{Cov}(|\mathcal{F}_{\sigma}\widetilde{W \rho}|^2) + \frac{i \hbar}{2}J \right) \ge 0
\end{array}
\label{eqIntroduction8}
\end{equation}
for Wigner distributions $W \rho$ belonging to some appropriate
maximal functional space and where
\begin{equation}
\mathcal{P} \left[W \rho \right]:=(2 \pi \hbar)^n ||W \rho||_{L^2(\mathbb{R}^{2n})}^2
 \label{eqIntroduction9}
\end{equation}
is the so-called purity of the state $\rho$. As before, we have defined:
\begin{equation}
\widetilde{W \rho } (z) := \frac{W \rho (z)}{||W \rho||_{L^2 (\mathbb{R}^{2n})}} , \hspace{1 cm} \mathcal{F}_{\sigma} \widetilde{W \rho } (\zeta) := \frac{\mathcal{F}_{\sigma} W\rho (\zeta )}{||W \rho||_{L^2 (\mathbb{R}^{2n})}}
 \label{eqIntroduction9.1}
\end{equation}
to make sure that $|\widetilde{W \rho } (z) |^2$ and $|\mathcal{F}_{\sigma} \widetilde{W \rho } (\zeta)|^2$ are properly normalized probability densities.

Moreover, we will also show that the first inequality in (\ref{eqIntroduction8}) becomes an equality if an only if the
state is pure.

So, in fact, the refined RSUP amounts to two inequalities. The
first inequality is
\begin{equation}
 \operatorname*{Cov}(|\widetilde{W \rho}|^2)+ \frac{1}{4} \operatorname*{Cov}(|\mathcal{F}_{\sigma}\widetilde{W \rho}|^2) + \frac{i \hbar}{2} J\ge 0.
\label{eqIntroduction9}
\end{equation}
In other words, the matrix $\operatorname*{Cov}(|\widetilde{W
\rho}|^2)+\frac{1}{4} \operatorname*{Cov}(|\mathcal{F}_{\sigma}\widetilde{W
\rho}|^2)$ also obeys the RSUP. The second inequality is
\begin{equation}
 \begin{array}{c}
 \operatorname*{Cov}(W \rho) + \frac{i \hbar}{2} J \ge \\
  \\
  \ge \mathcal{P} \left[W \rho \right] \left[ \operatorname*{Cov}(|\widetilde{W \rho}|^2)+ \frac{1}{4} \operatorname*{Cov}(|\mathcal{F}_{\sigma}\widetilde{W \rho}|^2) + \frac{i \hbar}{2} J \right] .
\end{array}
\label{eqIntroduction10}
\end{equation}
We notice that (\ref{eqIntroduction9}) and
(\ref{eqIntroduction10}) immediately imply the RSUP
(\ref{eqIntroduction4}).

Let us point out the main properties of the refined RSUP:

\vspace{0.3 cm}
\noindent
{\bf (1)} It is parsimonious, in the sense that it is a computable test as the RSUP, but not a complicated one as sets of necessary and sufficient conditions such as the Kastler, Loupias, Miracle-Sole (KLM) conditions \cite{Kastler,LouMiracle1,LouMiracle2}. In fact, we only have to compute the covariance matrices of $W \rho$, $|\widetilde{W \rho}|^2$ and $|\mathcal{F}_{\sigma} (\widetilde{W \rho})|^2$ and check inequalities (\ref{eqIntroduction8}).

\vspace{0.3 cm}
\noindent
{\bf (2)} It is invariant under linear symplectic and anti-symplectic transformations (see Theorem \ref{TheoremSymplecticCovariance}).

\vspace{0.3 cm} \noindent {\bf (3)} It makes a direct connection
with harmonic analysis, as it amounts to an inequality relating $W
\rho$ and its Fourier transform $\mathcal{F}_{\sigma}(W \rho)$.
Here we use the squares $|\widetilde{W\rho}|^2$ and $|\mathcal{F}_{\sigma}(\widetilde{W
\rho})|^2$, and so we are treating $W \rho$ as a wave function in
ordinary quantum mechanics on a $2n$-dimensional configuration
space\footnote{In this interpretation $\operatorname*{Cov}
(|W\rho|^2)$ is the covariance matrix of the $2n$ configurational
variables; and $\operatorname*{Cov}(|\mathcal{F}_{\sigma}W
\rho|^2)$ is the covariance matrix of the $2n$ conjugate
momenta.}.

\vspace{0.3 cm}
\noindent
{\bf (4)} It includes a pure state
condition. Indeed, inequality (\ref{eqIntroduction10}) is an
equality iff the state is pure.

\vspace{0.3 cm} \noindent {\bf (5)} It is stronger than the RSUP.
Indeed, inequality (\ref{eqIntroduction8}) implies immediately the
Robertson-Schr\"odinder uncertainty principle. Example
\ref{ExampleFinal1} shows that it is not equivalent to it.

\vspace{0.3 cm} \noindent {\bf (6)} It is a
deeper quantum mechanical requirement than the condition about a
minimal scale. For instance, in Example \ref{ExampleReview9}, we
show that the saturation (\ref{eqIntroduction12}) of the RSUP can
be easily achieved by many functions which are not Wigner
distributions. On the other hand, we prove in Theorem
\ref{TheoremSaturation} that the refined RSUP is saturated (i.e.
(\ref{eqIntroduction8}) and the saturation condition
(\ref{eqIntroduction12}) are both satisfied) if and only if the
state is a pure Gaussian Wigner function.

As a by-product of the refined RSUP, we also obtain a refinement of the Shannon and Hirschman inequalities \cite{Hirschman,Shannon} for Wigner distributions.

A famous theorem by Shannon \cite{Folland2,Shannon} states that if a
probability density
\begin{equation}
\mu(x) \ge0 , \hspace{1 cm} \int_{\mathbb{R}^{n}} \mu(x) dx=1,
\label{eqentropy1}%
\end{equation}
has finite covariance matrix $Cov(\mu)$, then its Boltzmann entropy
\begin{equation}
E(\mu) := - \int_{\mathbb{R}^{n}} \mu(x) \log\left(  \mu(x) \right)  dx
\label{eqentropy2}%
\end{equation}
is well defined and satisfies the inequality:
\begin{equation}
E(\mu) \le\frac{1}{2} \log\left[  (2 \pi e)^{n} \det\left(  Cov(\mu) \right)
\right]  .
\label{eqentropy3}%
\end{equation}

Another theorem due to Beckner \cite{Beckner}, Bialynicki-Birula and Mycielski \cite{Birula} and Hirschmann \cite{Hirschman} relates the entropy of
$|f|^{2}$, for $f\in L^{2}(\mathbb{R}^{n})$ and $||f||_{2}=1$ with that of
$|\mathcal{F}_{\hbar}f|^{2}$, where $(\mathcal{F}_{\hbar}f)$ is the $\hbar$-scaled Fourier transform. If the entropies of $|f|^{2}$ and $|\mathcal{F}_{\hbar
}f|^{2}$ are well defined then the Hirschman inequality reads:
\begin{equation}
\log\left(  \pi\hbar e\right)^{n}\leq E\left(  |f|^{2}\right)  +E\left(
|\mathcal{F}_{\hbar}f|^{2}\right)  .
\label{eqentropy5}%
\end{equation}
This inequality is sometimes called an entropic uncertainty principle as it
prevents a simultaneous sharp localization of $|f|^{2}$ and $|\mathcal{F}%
_{\hbar}f|^{2}$ and is saturated if and only if $f$ is a Gaussian with minimal
Heisenberg uncertainty.

Of course we may combine (\ref{eqentropy3}) and (\ref{eqentropy5})
and obtain the naive double inequality:
\begin{equation}%
\begin{array}
[c]{c}%
\log\left(  \pi\hbar e \right)^{n} \le E \left(  |f|^{2} \right)  + E
\left(  | \mathcal{F}_{\hbar} f|^{2} \right)  \le\\
\\
\le\log\left[  (2 \pi e)^{n} \sqrt{\det\left(  Cov (|f|^{2}) \right)
\cdot\det\left(  Cov (|\mathcal{F}_{\hbar} f|^{2})\right)  } \right]~.
\end{array}
\label{eqentropy6}%
\end{equation}
This can be stated in the following terms: if $|f|^{2} $ and $| \mathcal{F}%
_{\hbar}f|^{2}$ have finite covariance matrices, then they have well defined
entropies and inequality (\ref{eqentropy6}) holds. Moreover, we have
equalities throughout if and only if $f$ is a Gaussian. The inequality between
the first and the last term is, upon exponentiation, the Heinig-Smith
uncertainty principle \cite{Heinig}.

As a consequence of inequality (\ref{eqentropy6}) for the Wigner distribution and the refined RSUP (\ref{eqIntroduction8}), we derive the following Hirschman-Shannon inequality (Theorem \ref{Theorementropy1}):
\begin{equation}%
\begin{array}
[c]{c}%
\log\left[  (2\pi e)^{2n}\det\left(  Cov(W\rho)\right)  \right]  \geq\\
\\
\geq\log\left[  \left(  \pi e\mathcal{P}\left[W\rho\right]\right)  ^{2n}\sqrt{\det\left(
Cov(|\widetilde{W\rho}|^{2})\right)  \cdot\det\left(  Cov(|\mathcal{F}_{\hbar
}\widetilde{W\rho}|^{2})\right)  }\right]  \geq\\
\\
\geq2n\log\left(  \mathcal{P}\left[W\rho \right]\right)  +E\left(  |\widetilde{W\rho}%
|^{2}\right)  +E\left(  |\mathcal{F}_{\hbar}\widetilde{W\rho}|^{2}\right)
\geq\log\left(  \pi\hbar e\mathcal{P}\left[W\rho \right]\right)  ^{2n} .%
\end{array}
\label{eqentropy26}%
\end{equation}

We obtain an inequality throughout (\ref{eqentropy26}) if and only if $W \rho$ is the Wigner distribution of a pure Gaussian state.

For pure states $W \rho=W \psi$, the refined RSUP leads to the following Hirschman-Lieb-Shannon relation which
involves $W \psi$ only and not its Fourier transform (Corollary \ref{Corollary3}):
\begin{equation}%
\begin{array}
[c]{c}%
\log\left[  (2\pi e)^{n} \sqrt{\det\left(  Cov (W \psi) \right)  }\right]
\ge\\
\\
\ge\log\left[  \left(  2\pi e \right)  ^{n} \sqrt{\det\left(  Cov
(|\widetilde{W \psi}|^{2}) \right)  } \right]  \ge\\
\\
\ge E \left(  |\widetilde{W \psi}|^{2}\right)  \ge\log\left(  \frac{ \pi\hbar
e}{2} \right)  ^{2n}. %
\end{array}
\label{eqentropy30}%
\end{equation}

Before we conclude the introduction, let us comment on the new parts of the inequalities (\ref{eqentropy26},\ref{eqentropy30}). In (\ref{eqentropy26}) the last inequality is the Hirschman inequality for the Wigner function and the penultimate inequality is the Shannon inequality applied both to $|\widetilde{W \rho}|^2$ and to $|\mathcal{F}_{\hbar} \widetilde{W \rho}|^2$. The new inequalities are:
\begin{equation}
\det\left(  Cov(W\rho)\right) \ge \left(\frac{ \mathcal{P}\left[W\rho \right]}{2} \right)^{2n} \sqrt{\det\left(
Cov(|\widetilde{W\rho}|^{2})\right)  \cdot\det\left(  Cov(|\mathcal{F}_{\hbar
}\widetilde{W\rho}|^{2})\right)  },
\label{eqentropy30.1}%
\end{equation}
and
\begin{equation}
\log\left[  \left(\frac{2\pi e}{\mathcal{P}\left[W\rho\right]} \right)^{2n}\det\left(  Cov(W\rho)\right)  \right]  \geq\\
\\
E\left(  |\widetilde{W\rho}|^{2}\right)  +E\left(  |\mathcal{F}_{\hbar}\widetilde{W\rho}|^{2}\right).
\label{eqentropy30.2}%
\end{equation}

In (\ref{eqentropy30}) the last inequality is the entropic inequality of Lieb \cite{Lieb}. The penultimate inequality is the Shannon inequality applied to $|\widetilde{W \psi}|^{2}$. The new inequalities are:
\begin{equation}
\det\left(  Cov (W \psi) \right)  \ge \det\left(  Cov (|\widetilde{W \psi}|^2) \right),
\label{eqentropy30.3}%
\end{equation}
and
\begin{equation}
\log\left[  (2\pi e)^{n} \sqrt{\det\left(  Cov (W \psi) \right)  }\right]
\ge E \left(  |\widetilde{W \psi}|^{2}\right) .
\label{eqentropy30.4}%
\end{equation}

\section*{Notation}

The Plancherel-Fourier transform of a function $f \in L^1 (\mathbb{R}^n) \cap L^2 (\mathbb{R}^n)$ is defined by:
\begin{equation}
(\mathcal{F}f) (\omega):= \int_{\mathbb{R}^n} f(x) e^{-2 i \pi \omega \cdot x} dx
\label{eqNotation1}
\end{equation}
and the $\hbar$-scaled Fourier transform is:
\begin{equation}
(\mathcal{F}_{\hbar}f)(p):=\left(  \tfrac{1}{2\pi\hbar}\right)^{n/2}%
\int_{\mathbb{R}^{n}}f(x)e^{-\frac{i}{\hbar}x\cdot p}dx .
\label{eqentropy4}%
\end{equation}
We use lower case letters $f,g,\cdots$ for functions defined on the configuration space $\mathbb{R}^n$ and upper case letters from the middle of the alphabet $F,G, \cdots$ for functions on the phase space $\mathbb{R}^{2n}$. We shall use the physicists' convention for the inner product (anti-linear in the first argument and linear in the second)
\begin{equation}
(f|g)= \int_{\mathbb{R}^n} \overline{f(x)} g(x) dx.
\label{eqNotation2}
\end{equation}
To avoid a proliferation of subscripts, we use the notation
\begin{equation}
((F|G))= \int_{\mathbb{R}^{2n}} \overline{F(z)} G(z) dz
\label{eqNotation3}
\end{equation}
for the inner product on the phase space. Similarly we denote by $|| \cdot ||$ the norm on $L^2 (\mathbb{R}^n)$ and by  $||| \cdot |||$ that on $L^2 (\mathbb{R}^{2n})$. Sometimes, when more general $L^p $ norms are needed, we will be more specific and write $|| \cdot||_{L^p(\mathbb{R}^n)}$.

The Schwartz class of test functions is $\mathcal{S} (\mathbb{R}^n)$ and its dual - the space of tempered distributions - is denoted by $\mathcal{S}^{\prime} (\mathbb{R}^n)$. The distributional bracket is written $< \cdot, \cdot>$.

Given a functional space $L$, we denote by $\mathcal{F}L$ the set of distributions $f \in \mathcal{S}^{\prime} (\mathbb{R}^n)$ for which $\mathcal{F}f \in L$.

\section{A review of Wigner distributions}

In this section, we recapitulate the main results about Wigner distributions, which we will need in the sequel.

\subsection{Symplectic geometry}

The standard symplectic form on $\mathbb{R}^{2n} = \mathbb{R}_x^n \times \mathbb{R}_p^n$ is given by
\begin{equation}
\sigma (z,z^{\prime}) = z \cdot J^T z^{\prime}= p \cdot x^{\prime}- x \cdot p^{\prime},
\label{eqReview1}
\end{equation}
for $z=(x,p)$ and $z=(x^{\prime},p^{\prime})$. A linear automorphism $s: \mathbb{R}^{2n} \to \mathbb{R}^{2n}$ is a symplectic transformation if $\sigma (s(z),s(z^{\prime}))= \sigma (z,z^{\prime})$ for all $z,z^{\prime} \in \mathbb{R}^{2n}$. Let the symplectic transformation be represented by the matrix $S \in Gl(2n)$: $s(z) =Sz$. Then
\begin{equation}
S^T J  S=J.
\label{eqReview2}
\end{equation}
The set of real $2n \times 2n$ matrices which satisfy (\ref{eqReview2}) form a group called the symplectic group $Sp(n)$. If a matrix $A \in Gl(2n)$ is such that
\begin{equation}
A^T J  A= -J,
\label{eqReview3}
\end{equation}
then it is said to be anti-symplectic. Every anti-symplectic matrix $A$ can be written as \cite{PAMS}
\begin{equation}
A= T S,
\label{eqReview4}
\end{equation}
where $S \in Sp(n)$, and $T$ is usually interpreted as a "time-reversal" operator, since it amounts to a reversal of the particle's momentum:
\begin{equation}
T=\left(
\begin{array}{c c}
I & 0\\
0 & -I
\end{array}
\right).
\label{eqReview5}
\end{equation}
We shall denote the group of matrices which are either symplectic or anti-symplectic by $ASp(n)$.

Given a real symmetric positive definite matrix $B$ its symplectic eigenvalues (also called Williamson invariants) are given by the moduli of the eigenvalues of the matrix $B J^{-1}$ \cite{Gosson,Williamson}. Since they come in pairs $\pm i \lambda$ $(\lambda >0)$, we denote the $n$ moduli in increasing order by:
\begin{equation}
0 < \lambda_{\sigma,1}(B) \le \lambda_{\sigma,2}(B) \le \cdots \le \lambda_{\sigma,n}(B).
\label{eqReview6}
\end{equation}
The set
\begin{equation}
Spec_{\sigma} (B)= \left(\lambda_{\sigma,1}(B) , \lambda_{\sigma,2}(B) , \cdots , \lambda_{\sigma,n}(B) \right)
\label{eqReview7}
\end{equation}
is called the symplectic spectrum of $B$. Williamson's Theorem \cite{Williamson} states that the matrix $B$ can be diagonalized to a "normal" form by way of a similarity transformation with a symplectic matrix. More specifically, there exists $S \in Sp(n)$ such that
\begin{equation}
SBS^T = \left(
\begin{array}{c c}
\Lambda & 0 \\
0 & \Lambda
\end{array}
\right),
\label{eqReview8}
\end{equation}
where $\Lambda = diag \left(\lambda_{\sigma,1}(B) , \lambda_{\sigma,2}(B) , \cdots , \lambda_{\sigma,n}(B) \right)$.

\subsection{Weyl operators}

The symplectic Fourier transform of a function $F \in L^1 (\mathbb{R}^{2n}) \cap L^2 (\mathbb{R}^{2n})$ is given by:
\begin{equation}
(\mathcal{F}_{\sigma} F) (\zeta) = \frac{1}{(2 \pi \hbar)^n} \int_{\mathbb{R}^{2n}} F(z) e^{-\frac{i}{\hbar} \sigma (\zeta,z)} dz.
\label{eqReview9}
\end{equation}
It is related to the Fourier transform (\ref{eqNotation1}) and the $\hbar$-scaled Fourier transform (\ref{eqentropy4}) by:
\begin{equation}
(\mathcal{F}_{\sigma} F)(\zeta) = \frac{1}{(2 \pi \hbar)^n}(\mathcal{F} F) \left(  \frac{J \zeta}{2 \pi \hbar}\right)=\left(\mathcal{F}_{\hbar} F \right) (J \zeta).
\label{eqReview10}
\end{equation}
The symplectic Fourier transform is an involution which extends by duality to an involutive automorphism $\mathcal{S}^{\prime} (\mathbb{R}^{2n}) \to \mathcal{S}^{\prime} (\mathbb{R}^{2n})$.

Given a symbol $a \in \mathcal{S}^{\prime} (\mathbb{R}^{2n})$, the associated Weyl operator is given by the Bochner integral \cite{Birk,Gosson}:
\begin{equation}
\widehat{A}:= \left( \frac{1}{2 \pi \hbar} \right)^n \int_{\mathbb{R}^{2n}} (\mathcal{F}_{\sigma} a) (z_0) \widehat{T} (z_0 ) dz_0,
\label{eqReview11}
\end{equation}
where $\widehat{T} (z_0)$ is the Heisenberg-Weyl operator
\begin{equation}
(\widehat{T} (z_0) f) (x)= e ^{\frac{i}{\hbar} p_0 \cdot\left(x- \frac{x_0}{2} \right)} f(x-x_0),
\label{eqReview12}
\end{equation}
for $z_0 = (x_0,p_0) \in \mathbb{R}^{2n}$ and $f \in \mathcal{S} (\mathbb{R}^n)$. We remark that the operator $\widehat{A}$ is formally self-adjoint if and only its symbol $a$ is real.

The Weyl correspondence, written $a \overset{\mathrm{Weyl}}{\longleftrightarrow} \widehat{A}$ or $\widehat{A} \overset{\mathrm{Weyl}}{\longleftrightarrow} a$, between an element $a \in \mathcal{S}^{\prime} (\mathbb{R}^{2n})$ and the Weyl operator it defines is bijective; in fact the Weyl transformation is one-to-one from $\mathcal{S}^{\prime} (\mathbb{R}^{2n})$ onto the space $\mathcal{L}\left( \mathcal{S}(\mathbb{R}^{n}),\mathcal{S}^{\prime} (\mathbb{R}^{n})\right)$ of linear continuous maps $\mathcal{S}(\mathbb{R}^{n}) \to \mathcal{S}^{\prime} (\mathbb{R}^{n})$ (see e.g. Maillard \cite{Maillard}, Unterberger \cite{Unterberger} or Wong \cite{Wong}). This can be proven using Schwartz's kernel theorem and the fact that the Weyl symbol $a$ of the operator $\widehat{A}$ is related to the distributional kernel $K_A$ of that operator by the partial Fourier transform with respect to the y variable
\begin{equation}
a(x, p) = \int_{\mathbb{R}^n} K_A \left( x+ \frac{y}{2},x- \frac{y}{2} \right) e^{- \frac{i}{\hbar}p \cdot y} dy,
\label{eqReview13}
\end{equation}
where $K_A \in  \mathcal{S}^{\prime} (\mathbb{R}^n \times \mathbb{R}^n )$ and the Fourier transform is defined in the usual distributional sense. Conversely, the kernel $K_A$ is expressed in terms of the symbol $a$ by the inverse Fourier transform
\begin{equation}
K_A(x, y) = \left( \frac{1}{2 \pi \hbar} \right)^n \int_{\mathbb{R}^n} a \left(\frac{x+y}{2},p \right) e^{ \frac{i}{\hbar} p\cdot (x-y)} dp.
\label{eqReview14}
\end{equation}

Weyl operators enjoy the following symplectic covariance property \cite{Folland,Birk,Gosson,Gro,Wong}. Let $S \in Sp(n)$ and $\widehat{S} \in Mp(n)$ be one of the two metaplectic operators that project onto $S$. Recall that metaplectic operators constitute a unitary representation of the two-fold cover $Sp_2(n)$ of $Sp(n)$. If $\widehat{A}: \mathcal{S} (\mathbb{R}^n) \to \mathcal{S}^{\prime} (\mathbb{R}^n)$ is  a Weyl operator with symbol $a \in  \mathcal{S}^{\prime} (\mathbb{R}^{2n})$, then we have
\begin{equation}
\widehat{S}^{-1} \widehat{A} \widehat{S} \overset{\mathrm{Weyl}}{\longleftrightarrow} a \circ S .
\label{eqReview14.1}
\end{equation}
Since an anti-symplectic transformation is the composition $TS$ (see (\ref{eqReview4})) it suffices to consider the action of $T$. Quantum mechanically, this is implemented by the anti-linear operator
\begin{equation}
(\widehat{T}f)(x) = \overline{f(x)}.
\label{eqReview14.1.A}
\end{equation}
This also supports the interpretation of $T$ as a time reversal. If $f$ obeys the Schr\"odinger equation, then $\overline{f}$ obeys the same equation with the time reversal $t \to -t$.

Assuming that the product $\widehat{A}\widehat{B}$ exists (which is the case for instance if $\widehat{B} : \mathcal{S}(\mathbb{R}^{n}) \to \mathcal{S}(\mathbb{R}^{n})$) the Weyl symbol $c$
of $\widehat{C}= \widehat{A}\widehat{B}$ and its symplectic Fourier transform $\mathcal{F}_{\sigma} c$ are given by the formulae:
\begin{equation}
c(z) =
\left(\frac{1}{4 \pi \hbar} \right)^{2n} \int_{\mathbb{R}^{2n}} \int_{\mathbb{R}^{2n}} a\left(z + \frac{u}{2} \right) b \left(z - \frac{v}{2} \right) e^{\frac{i}{2 \hbar} \sigma (u,v)} du dv,
\label{eqReview15}
\end{equation}
and
\begin{equation}
(\mathcal{F}_{\sigma} c)(z) =\left(\frac{1}{2 \pi \hbar} \right)^{n} \int_{\mathbb{R}^{2n}} (\mathcal{F}_{\sigma} a) (z-z^{\prime}) (\mathcal{F}_{\sigma} b) (z^{\prime}) e^{\frac{i}{2 \hbar} \sigma (z,z^{\prime})} d z^{\prime}.
\label{eqReview16}
\end{equation}
The first formula is often written $c = a \star_{\hbar} b$ and $a \star_{\hbar} b$ is called the \textit{twisted product} or \textit{Moyal product} (see e.g. \cite{Folland,Groenewold,Moyal,Wong}).

\subsection{Quantum states and Wigner functions}

An important case consists of rank one operators of the form:
\begin{equation}
\left(\widehat{\rho}_{f,g} h \right) (x) = (g|h) f (x),
\label{eqReview17}
\end{equation}
for fixed $f,g \in L^2 (\mathbb{R}^n)$ acting on $h \in L^2 (\mathbb{R}^n)$. They are Hilbert-Schmidt operators with kernel $K_{f,g}(x,y) = (f \otimes \overline{g}) (x,y)=f(x)  \overline{g(y)} $. According to (\ref{eqReview13}), the associated Weyl symbol is:
\begin{equation}
\rho_{f,g}(x,p) = \int_{\mathbb{R}^n} f \left( x + \frac{y}{2} \right) \overline{g \left( x - \frac{y}{2} \right)} e^{- \frac{i}{\hbar} p \cdot y} dy.
\label{eqReview18}
\end{equation}
This is just the cross-Wigner function up to a multiplicative constant:
\begin{equation}
\begin{array}{c}
W(f,g)(x,p)= \left(\frac{1}{2 \pi \hbar} \right)^n  \rho_{f,g} (x,p)=\\
 \\
 = \left(\frac{1}{2 \pi \hbar} \right)^n \int_{\mathbb{R}^n} f \left( x + \frac{y}{2} \right) \overline{g \left( x - \frac{y}{2} \right)} e^{- \frac{i}{\hbar} p \cdot y} dy.
\end{array}
\label{eqReview19}
\end{equation}
From (\ref{eqReview14.1}), we conclude that
\begin{equation}
W(\widehat{S}f,\widehat{S}g)(z)=W(f,g)(S^{-1} z).
\label{eqReview19.1}
\end{equation}
If $g=f$, we simply write $Wf$ meaning $W(f,f)$:
\begin{equation}
Wf(x,p)= \left(\frac{1}{2 \pi \hbar} \right)^n \int_{\mathbb{R}^n} f \left( x + \frac{y}{2} \right) \overline{f \left( x - \frac{y}{2} \right)} e^{- \frac{i}{\hbar} p \cdot y} dy.
\label{eqReview20}
\end{equation}
We say that $W f$ is the Wigner function \cite{Wigner} associated with the pure state $f \in L^2 (\mathbb{R}^n)$.

In quantum mechanics, one usually has to deal with statistical mixtures of pure states. This means that pure states represented by the rank one operators $\widehat{\rho}_f =\widehat{\rho}_{f,f}$ (see (\ref{eqReview17})) are replaced by convex combinations of the form:
\begin{equation}
\widehat{\rho} = \sum_{\alpha} p_{\alpha} \widehat{\rho}_{f_{\alpha}},
\label{eqReview21}
\end{equation}
with $p_{\alpha} \ge 0$ and $\sum_{\alpha} p_{\alpha} =1$. The convergence of the series in (\ref{eqReview21}) is understood in the sense of the trace norm. Operators of this form are called density matrices. They are positive trace class operators with unit trace. The set of density matrices - the set of states - is denoted by $\mathcal{S}( L^2 (\mathbb{R}^n))$. A density matrix $\widehat{\rho}$ is a Hilbert-Schmidt operator with kernel:
\begin{equation}
\rho (x,y) = \sum_{\alpha} p_{\alpha} f_{\alpha}(x) \overline{ f_{\alpha}(y)}.
\label{eqReview22}
\end{equation}
The associated Wigner function is
\begin{equation}
\begin{array}{c}
W \rho (x,p) = \sum_{\alpha} p_{\alpha} Wf_{\alpha} (x,p)= \left(\frac{1}{2 \pi \hbar}\right)^n \int_{\mathbb{R}^n} \rho \left(x+ \frac{y}{2},x- \frac{y}{2} \right) e^{- \frac{i}{\hbar} p \cdot y} dy=\\
\\
=  \left(\frac{1}{2 \pi \hbar}\right)^n \sum_{\alpha} p_{\alpha} \int_{\mathbb{R}^n}  f_{\alpha}\left(x+ \frac{y}{2} \right) \overline{f_{\alpha}\left(x- \frac{y}{2} \right)} e^{- \frac{i}{\hbar} p \cdot y} dy
\end{array}
\label{eqReview23}
\end{equation}
with uniform convergence.

We shall denote by $\mathcal{W} (\mathbb{R}^{2n})$ the set of all Wigner functions associated with density matrices, that is the range of the Weyl transform acting on $\mathcal{S}(L^2 (\mathbb{R}^n))$. This is basically the set of quantum mechanical states in the Weyl-Wigner representation. One can tell whether an element $W \rho \in \mathcal{W} (\mathbb{R}^{2n})$ represents a pure or a mixed state by calculating its purity:
\begin{equation}
\mathcal{P} \left[W \rho \right]:= (2 \pi \hbar)^n ||| W \rho|||^2.
\label{eqReview23.1}
\end{equation}
We have:
\begin{equation}
\left\{
\begin{array}{l l}
\mathcal{P} \left[W \rho \right]=1, & \mbox{if $W \rho$ is a pure state}\\
& \\
\mathcal{P} \left[W \rho \right]<1, & \mbox{if $W \rho$ is a mixed state}
\end{array}
\right.
\label{eqReview23.1}
\end{equation}
One aspect which makes the Wigner formalism very appealing is the fact that expectation values are computed with a formula akin to classical statistical mechanics \cite{Folland,Gosson,Wong}. Indeed, if $\widehat{A}$ is a self-adjoint Weyl operator with symbol $a \in \mathcal{S} (\mathbb{R}^{2n})$, then it can be shown that
\begin{equation}
(g|\widehat{A} f)= ((a| W(g,f) )),
\label{eqReview24}
\end{equation}
for $f,g \in  \mathcal{S} (\mathbb{R}^n)$. In particular, we have:
\begin{equation}
<\widehat{A} >_f=(f|\widehat{A} f)= \int_{\mathbb{R}^{2n}} a(x,p) W f(x,p) dx dp.
\label{eqReview25}
\end{equation}
For a generic self-adjoint Weyl operator $\widehat{A} \overset{\mathrm{Weyl}}{\longleftrightarrow} a$ which is also trace-class, the following identity holds:
\begin{equation}
Tr (\widehat{A}) =  \left( \frac{1}{2 \pi \hbar}\right)^n \int_{\mathbb{R}^{2n}} a(z) dz.
\label{eqReview26}
\end{equation}
If $\widehat{A} \overset{\mathrm{Weyl}}{\longleftrightarrow} a$ and $\widehat{B} \overset{\mathrm{Weyl}}{\longleftrightarrow} b$ are Weyl operators such that $\widehat{A} \widehat{B}$ is trace-class, then we have\cite{Birk,Gosson}:
\begin{equation}
Tr (\widehat{A}\widehat{B}) =   \left( \frac{1}{2 \pi \hbar}\right)^n \int_{\mathbb{R}^{2n}} a(z)\star_{\hbar} b(z) dz = \left( \frac{1}{2 \pi \hbar}\right)^n \int_{\mathbb{R}^{2n}} a(z) b(z) dz .
\label{eqReview27}
\end{equation}
In particular, for density matrices (\ref{eqReview25}) generalizes to
\begin{equation}
< \widehat{A}>_{\rho} = Tr ( \widehat{A} \widehat{\rho}) =\int_{\mathbb{R}^{2n}} a(z) W \rho (z) dz,
\label{eqReview27.1}
\end{equation}
provided $\widehat{A} \widehat{\rho}$ is trace class.

In general, it is very difficult to determine whether a given phase space function $F$ is the Wigner function of some density matrix $\widehat{\rho} \in \mathcal{S}(L^2 (\mathbb{R}^n))$. It can be shown that \cite{Dias1,Lions}:
\begin{theorem}\label{TheoremReview1}
Let $F: \mathbb{R}^{2n} \to \mathbb{C}$ be a measurable function. We have $F \in \mathcal{W} (\mathbb{R}^{2n})$ if and only if:

\vspace{0.3 cm}
\noindent
(i) $F$ is a real function,

\vspace{0.3 cm}
\noindent
(ii) $F \in L^2 (\mathbb{R}^{2n})$,

\vspace{0.3 cm}
\noindent
(iii) $\int_{\mathbb{R}^{2n}}F (z) dz =1$,

\vspace{0.3 cm}
\noindent
(iv) $\int_{\mathbb{R}^{2n}}F (z) W f (z) dz \ge 0$, for all $f \in L^2 (\mathbb{R}^n) $.
\end{theorem}
The first two conditions mean that $F$ is the Weyl symbol of a self-adjoint Hilbert-Schmidt operator. The last condition means that this operator is positive. These conditions, together with (iii), imply that the operator is trace class and that the trace is equal to one.

This set of conditions are somewhat tautological as they require the knowledge of the set of pure state Wigner functions $W f$ to check the positivity (iv).

There are an alternative set of necessary and sufficient conditions, the Kastler, Loupias, Miracle-Sole (KLM) conditions \cite{Kastler,LouMiracle1,LouMiracle2}, that do not share this disadvantage. However, they are virtually impossible to check, as they amount to verifying the positivity of an infinite hierarchy of matrices of growing dimension (see also \cite{Dias2,Narcow2,Narcow3,Narconnell}). In practise, these conditions can be checked up to a given finite order, in which case they provide a set of necessary but not sufficient conditions for a measurable function $F$ to belong to $\mathcal{W}(\mathbb{R}^{2n})$. Other, more practical, necessary conditions are the uncertainty principles.

\subsection{Uncertainty principles}

One of the hallmarks of quantum mechanics is the uncertainty principle. For a survey of mathematical aspects of the uncertainty principle see \cite{Folland2}. Good discussions on the physical interpretation and implications of the uncertainty principle can be found in \cite{Busch1,Busch2}. Roughly speaking, an uncertainty principle poses an obstruction to a state being sharply localized both in position and in momentum space. There are various ways one can formulate this principle mathematically. For instance, one can show that (see e.g. \cite{PLA,Janssen})
\begin{theorem}\label{TheoremReview2}
If $W \rho \in \mathcal{W}(\mathbb{R}^{2n})$, then $W \rho$ is uniformly continuous and it cannot be compactly supported.
\end{theorem}
Other results for the support of joint position-momentum (or time-frequency) representations can be found in \cite{Demange} for the ambiguity function and in \cite{Wilczek} for the short-time Fourier transform. The continuous wavelet transform, which is a time-scale representation, was also shown to have non-compact support in \cite{Wilczek}. Ghobber and Jaming \cite{Ghobber1,Ghobber2} derived uncertainty principles for arbitrary integral operators (Fourier, Dunkl, Clifford transforms, etc) which have bounded kernels and satisfy a Plancherel theorem. A sharp version of the Beurling uncertainty principle was proven by B. Demange for the ambiguity function \cite{Demange}.

The most famous version of an uncertainty principle is Heisenberg's uncertainty principle:
\begin{theorem}\label{TheoremReview3}
Let $< \widehat{X}_i>= Tr(\widehat{X}_i \widehat{\rho})$, $<\widehat{P}_i>= Tr(\widehat{P}_i \widehat{\rho})$, $\Delta x_i^2= Tr((\widehat{X}_i-< \widehat{X}_i>\widehat{I})^2 \widehat{\rho})$ and $\Delta p_i^2= Tr((\widehat{P}_i-< \widehat{P}_i>\widehat{I})^2 \widehat{\rho})$ denote the expectation values and the variances of the particle's position and momentum which we assume to be finite. Then:
\begin{equation}
\Delta x_i \Delta p_i \ge \frac{\hbar}{2},
\label{eqReview28}
\end{equation}
for $i=1, \cdots, n$.
\end{theorem}
This theorem does not take into account the correlations $x_ix_j$, $p_ip_j$ or $x_i p_j$. A first generalization would be the Heinig-Smith uncertainty principle \cite{Heinig}:
\begin{theorem}\label{TheoremReview4}
Let $f \in L^2 (\mathbb{R}^n)$, and $\widetilde{f}$ as before, such that
\begin{equation}
d_{ij}= \int_{\mathbb{R}^n} (x_i-<x_i>)(x_j-<x_j>)  |\widetilde{f}(x)|^2 dx
\label{eqReview29}
\end{equation}
and
\begin{equation}
\widetilde{d}_{ij}=   \int_{\mathbb{R}^n} (\omega_i-<\omega_i>)(\omega_j-<\omega_j>)  |(\mathcal{F} \widetilde{f})(\omega)|^2 d \omega
\label{eqReview30}
\end{equation}
are finite for all $i,j=1, \cdots, n$. Here
\begin{equation}
\begin{array}{l}
<x_i>= \int_{\mathbb{R}^n} x_i |\widetilde{f}(x)|^2 dx, \\
\\
<\omega_i>= \int_{\mathbb{R}^n} \omega_i |(\mathcal{F} \widetilde{f})(\omega)|^2 d \omega.
\end{array}
\label{eqReview31.1}
\end{equation}

Then the covariance matrices $D=(d_{ij})_{ij}$ and $\widetilde{D}=(\widetilde{d}_{ij})_{ij}$ satisfy:
\begin{equation}
(\det D)(\det \widetilde{D}) \ge \left(\frac{1}{4 \pi} \right)^{2n}.
\label{eqReview31}
\end{equation}
Moreover, an equality holds if and only if $f$ is a generalized Gaussian of the form:
\begin{equation}
f(x)=e^{- \pi x \cdot A x + 2 \pi b \cdot x + c},
\label{eqReview31.1}
\end{equation}
where $A \in Gl(n, \mathbb{C})$ is symmetric with $Re(A) >0$, and $b \in \mathbb{C}^n$, $c \in \mathbb{C}$.
\end{theorem}

\begin{remark}\label{RemarkReview5.1}
The previous theorem also holds for density matrices. Moreover, as in Theorem \ref{TheoremReview3}, we could have assumed immediately that $f$ is normalized $||f||=||\mathcal{F}f ||=1$. We have chosen this version here, because this is how we will need this result below.
\end{remark}

\begin{remark}\label{RemarkReview5}
It will be useful in the sequel to write the Heinig-Smith inequality for functions $F$ defined in the phase space $\mathbb{R}^{2n}$ and express it in terms of the symplectic Fourier transform. Thus, in view of (\ref{eqReview10}):
\begin{equation}
 \operatorname*{Cov} (|\mathcal{F} \widetilde{F}|^2)= \frac{1}{(2 \pi \hbar)^2} J \operatorname*{Cov} (|\mathcal{F}_{\sigma} \widetilde{F}|^2)
J^T .
\label{eqReview32.1}
\end{equation}
Replacing $D$ by $\operatorname*{Cov}(|\widetilde{F}|^2)$, $\widetilde{D}$ by $\operatorname*{Cov} (|\mathcal{F}_{\sigma} \widetilde{F}|^2)$ and $n$ by $2n$ in (\ref{eqReview31}) yields:
\begin{equation}
\det \left(\operatorname*{Cov}(|\widetilde{F}|^2) \right) \det \left(\operatorname*{Cov} (|\mathcal{F}_{\sigma} \widetilde{F}|^2) \right) \ge \left( \frac{\hbar}{2} \right)^{4n}.
\label{eqReview32.2}
\end{equation}
Moreover, the inequality (\ref{eqReview32.2}) becomes an equality if and only if $F$ is of the form:
\begin{equation}
F(z)=e^{- \pi z \cdot A z + 2 \pi b \cdot z + c},
\label{eqReview32.3}
\end{equation}
where $A \in Gl(2n, \mathbb{C})$ is symmetric with $Re(A) >0$, and $b \in \mathbb{C}^{2n}$, $c \in \mathbb{C}$.
\end{remark}

Other uncertainty principles involving quadratic forms were obtained by B. Demange \cite{Demange}.

Theorems \ref{TheoremReview3} and \ref{TheoremReview4} still do not account for the position-momentum correlations. A consequence of this is that they are not invariant under linear (anti-)symplectic transformations. On the other hand, the Robertson-Schr\"odinger uncertainty principle is symplectially invariant \cite{Gosson}.
\begin{theorem}\label{TheoremReview6}
{\bf (Robertson-Schr\"odinger uncertainty principle)} Let $ \operatorname*{Cov} (W \rho)$ be the covariance matrix of $W \rho$ (or $\widehat{\rho}$) with entries:
\begin{equation}
\operatorname*{Cov} (W \rho)= \int_{\mathbb{R}^{2n}} (z-<z>) (z-<z>)^T W \rho (z) dz ,
\label{eqReview33}
\end{equation}
which we assume to be finite. Then we have:
\begin{equation}
 \operatorname*{Cov} (W \rho) + \frac{i\hbar}{2} J \ge 0.
\label{eqReview34}
\end{equation}
That is, the matrix $\operatorname*{Cov} (W \rho) + \frac{i\hbar}{2} J$ is positive in $\mathbb{C}^{2n}$.
\end{theorem}
By diagonalizing $\operatorname*{Cov} (W \rho)$ with the help of Williamson's Theorem and using the symplectic invariance of (\ref{eqReview34}), we conclude that the RSUP is equivalent to \cite{Gosson,Narcow2,Narcow}
\begin{equation}
\lambda_{\sigma,1}\left( \operatorname*{Cov} (W \rho)\right) \ge \frac{\hbar}{2},
\label{eqReview35}
\end{equation}
where $\lambda_{\sigma,1}\left(\operatorname*{Cov} (W \rho)\right)$ is the smallest symplectic eigenvalue of $\operatorname*{Cov} (W \rho)$. The extremal situation
\begin{equation}
\lambda_{\sigma,1}\left(\operatorname*{Cov} (W \rho)\right)=\lambda_{\sigma,2}\left(\operatorname*{Cov} (W \rho)\right)= \cdots=\lambda_{ \sigma,n}\left(\operatorname*{Cov} (W \rho)\right)=\frac{\hbar}{2},
\label{eqReview36}
\end{equation}
corresponds to a minimal uncertainty density matrix. In $\mathcal{W}(\mathbb{R}^{2n})$ this can only be achieved by Gaussian pure states \cite{Gosson}.

\begin{theorem}\label{TheoremReview7}
Let $\operatorname*{Cov} (W \rho)$ satisfy the RSUP (\ref{eqReview34}) with $W \rho \in \mathcal{W} (\mathbb{R}^{2n})$. Then it saturates the uncertainty principle in the sense of (\ref{eqReview36}) if and only if $W \rho =W f$ is the Wigner function of a Gaussian pure state $f$ of the form (\ref{eqReview31.1}).
\end{theorem}

\begin{remark}\label{RemarkLittlejohn}
The Wigner function of a Gaussian pure state (\ref{eqReview31.1}) can be expressed as
\begin{equation}
W f(z)= \frac{1}{( \pi \hbar)^n} \exp \left( - \frac{1}{2} (z-z_0) \cdot \left( \operatorname*{Cov} (Wf) \right)^{-1} (z-z_0)\right),
\label{eqReview37}
\end{equation}
where $z_0 \in \mathbb{R}^{2n}$ and the covariance matrix $ \operatorname*{Cov} (Wf)$ is a real symmetric positive-definite $2 n \times 2n$ matrix such that
\begin{equation}
\frac{2}{\hbar} \operatorname*{Cov} (Wf) \in Sp(n).
\label{eqReview38}
\end{equation}
This is known by physicists as Littlejohn's Theorem \cite{Littlejohn} but was first proven by Bastiaans \cite{Bastiaans}.
\end{remark}

Theorem \ref{TheoremReview7} is valid in $\mathcal{W}(\mathbb{R}^{2n})$ but not in $L^2(\mathbb{R}^{2n})$. In fact, the RSUP is only a necessary condition for a real phase space function $F$ to be a Wigner function. However, it is not sufficient (not even if saturated). Here is a counter-example.
\begin{example}\label{ExampleReview9}
Let $F$ be the function on $\mathbb{R}^2$ defined by
\begin{equation}
F(z)=\frac{1}{\pi R^2} \chi_R (z),
\label{eqReview39}
\end{equation}
where $\chi_R (z)$ is the indicator function of the disc of radius $R$ centered at the origin
\begin{equation}
\chi_R (z)= \left\{
\begin{array}{l l}
1 & if ~ |z| \le R\\
0 & if ~|z| >R
\end{array}
\right. ~.
\label{eqReview40}
\end{equation}
The function $F$ is real and normalized. However, it cannot possibly be a Wigner function, because it is discontinuous and because it has compact support. But, as we now show, it can nevertheless satisfy the Robertson-Schr\"odinger uncertainty principle, or even saturate it, provided we choose the radius $R$ appropriately.

A simple calculation shows that the covariance matrix of $F$ is
\begin{equation}
 \operatorname*{Cov} (F)= \frac{R^2}{4} I,
\label{eqReview41}
\end{equation}
where $I$ is the identity matrix. The Williamson invariant of $\operatorname*{Cov} (F)$ is
\begin{equation}
\lambda_{\sigma,1}(\operatorname*{Cov}(F)) = \frac{R^2}{4}.
\label{eqReview42}
\end{equation}
So the Robertson-Schr\"odinger uncertainty principle is satisfied, if and only if
\begin{equation}
R \ge \sqrt{2 \hbar},
\label{eqReview43}
\end{equation}
and saturated provided
\begin{equation}
R = \sqrt{2 \hbar}.
\label{eqReview44}
\end{equation}
In higher dimension $n >1$, we may consider the tensor products
\begin{equation}
F(z)= \Pi_{j=1}^n \frac{1}{\pi R^2} \chi_R (z_j).
\label{eqReview44}
\end{equation}
Again, if (\ref{eqReview43}) holds, then $F$ satisfies the RSUP and it saturates it for (\ref{eqReview44}).
\end{example}

Thus, as we argued in the introduction, the only imprint of quantum mechanics in the RSUP is a scale requirement related to Planck's constant. Indeed, we have the more dramatic result that, provided the covariance matrix is finite and positive-definite, then any phase space function satisfies the RSUP after a scale transformation.

\begin{lemma}\label{LemmaReview10}
Let $F: \mathbb{R}^{2n} \to \mathbb{R}$ be a normalized measurable function such that its covariance matrix $\operatorname*{Cov}(F)$ is finite and positive-definite. Then there exists $0<\mu \le 1 $ such that $F_{\mu}(z)=  \mu^{2n} F(\mu z)$ satisfies the RSUP.
\end{lemma}

\begin{proof}
Let $\lambda_{\sigma,1} \left(  \operatorname*{Cov}(F)\right)$ denote the smallest Williamson invariant of $\operatorname*{Cov} (F)$. If $\lambda_{\sigma,1} \left( \operatorname*{Cov}(F)\right) \ge \frac{\hbar}{2}$, we choose $\mu=1$ and we are done. Alternatively, suppose that $\lambda_{\sigma,1} \left( \operatorname*{Cov} (F)\right) < \frac{\hbar}{2}$. Since $\operatorname*{Cov} (F_{\mu}) = \frac{1}{\mu^2} \operatorname*{Cov} (F)$, we conclude that $\lambda_{\sigma,1} \left( \operatorname*{Cov} (F_{\mu})\right)= \frac{ \lambda_{\sigma,1} \left( \operatorname*{Cov} (F)\right)}{\mu^2}$. If we choose
\begin{equation}
0< \mu < \sqrt{\frac{2 \lambda_{\sigma,1} \left( \operatorname*{Cov} (F)\right)}{\hbar}} <1,
\label{eqReview45}
\end{equation}
then $F_{\mu}$ satisfies the RSUP.
\end{proof}

\subsection{Modulation spaces}

To conclude this section, we address the question of finiteness of the covariance matrix elements of a given function. The proper setting in this respect is that of Feichtinger's modulation spaces \cite{Hans1,fei81}\footnote{For a detailed review see \cite{Gro}; we are using here their formulation in terms of the Wigner
distribution as in \cite{Birk}.}. These are a class of functional spaces which, roughly speaking, describe the integrability, decay and smoothness properties of a function and its Fourier transform.

Let $\langle z\rangle=(1+|z|^2)^{1/2}$; we will
call $\langle\cdot\rangle$ the standard weight function. The modulation space
$M_{s}^{q}(\mathbb{R}^{n})$ consists of all distributions $f\in
\mathcal{S}^{\prime}(\mathbb{R}^{n})$ such that $W(f,g)\in L_{s}%
^{q}(\mathbb{R}^{2n})$ for all $g \in\mathcal{S}(\mathbb{R}^{n}) \backslash \left\{0 \right\}$. Here
$L_{s}^{q}(\mathbb{R}^{2n})$ is the space of all functions $F$ on
$\mathbb{R}^{2n}$ such that%
\begin{equation}
||F||_{L_{s}^{q}}=\left(  \int_{\mathbb{R}^{2n}}\left(  \langle z\rangle
^{s}|F(z)|\right)^{q}dz\right)  ^{1/q}<\infty.
\label{eqReview46}
\end{equation}
One shows that $M_{s}^{q}(\mathbb{R}^{n})$ is a Banach space for the norms%
\begin{equation}
||f||_{g,M_{s}^{q}}=||W(f,g)||_{L_{s}^{q}};
\label{eqReview47}
\end{equation}
these norms are in fact all equivalent for different choices of window $g$, so that the condition $f\in
M_{s}^{q}(\mathbb{R}^{n})$ holds if $W(f,g)\in L_{s}^{q}(\mathbb{R}%
^{2n})$ for one $g \in\mathcal{S}(\mathbb{R}^{n})\backslash \left\{0 \right\}$; even more surprisingly,
we have $f\in M_{s}^{q}(\mathbb{R}^{n})$ if and only if $Wf =W(f,f) \in L_{s}%
^{q}(\mathbb{R}^{2n})$ (but it is of course not immediately obvious from this
characterization that $M_{s}^{q}(\mathbb{R}^{n})$ is a vector space!). The
class of modulation spaces contain as particular cases several well-known
function spaces. For instance, the Shubin class%
\begin{equation}
Q^{s}(\mathbb{R}^{n})=L_{s}^{2}(\mathbb{R}^{n})\cap H^{s}(\mathbb{R}^{n}),
\label{eqReview48}
\end{equation}
corresponds to $M_{s}^{2} (\mathbb{R}^n)$. In particular, it can be shown that:
\begin{equation}
M_{1}^{2} (\mathbb{R}^n) \simeq \left\{f \in \mathcal{S}^{\prime} (\mathbb{R}^n): ~ \int_{\mathbb{R}^n} (1+|x|^2) \left(|f(x)|^2 + |(\mathcal{F}f)(x)|^2 \right) dx < \infty \right\}.
\label{eqReview49}
\end{equation}

The case $q=1$, $s=0$ is also noteworthy. The corresponding
modulation space $M_{0}^{1}(\mathbb{R}^{n})$ is called Feichtinger's algebra
and is usually denoted by $S_{0}(\mathbb{R}^{n})$. The Feichtinger algebra is
an algebra for both pointwise multiplication and convolution. One proves that
$S_{0}(\mathbb{R}^{n})$ is the smallest Banach space containing $\mathcal{S}%
(\mathbb{R}^{n})$ and which is invariant under the action of metaplectic
operators and translations. We have the inclusion%
\begin{equation}
S_{0}(\mathbb{R}^{n})\subset C^{0}(\mathbb{R}^{n})\cap L^{1}(\mathbb{R}%
^{n})\cap \mathcal{F} L^{1}(\mathbb{R}^{n}).
\label{eqReview50}
\end{equation}

The modulation spaces $M_{s}^{q}(\mathbb{R}^{n})$ have similar properties:

\begin{proposition}
(i) Each space $M_{s}^{q}(\mathbb{R}^{n})$ is invariant under the action of
the Heisenberg-Weyl operators $\widehat{T}(z)$ and there exists a constant $C>0$ such that%
\begin{equation}
||\widehat{T}(z)f||_{g,M_{s}^{q}}\leq C\langle z\rangle^{s}%
||f||_{g ,M_{s}^{q}};
\label{eqReview51}
\end{equation}
(ii) If $\widehat{S}\in\operatorname*{Mp}(n)$ and $f \in M_{s}%
^{q}(\mathbb{R}^{n})$ then $\widehat{S} f \in M_{s}^{q}(\mathbb{R}^{n})$;

\noindent
(iii) $\mathcal{S}(\mathbb{R}^{n})$ is dense in each of the spaces $M_{s}%
^{q}(\mathbb{R}^{n})$ and we have%
\begin{equation}
\mathcal{S}(\mathbb{R}^{n})=\cap_{s\geq0}M_{s}^{2}(\mathbb{R}^{n}).
\label{eqReview52}
\end{equation}
\end{proposition}

We remark that the Feichtinger algebra $S_{0}(\mathbb{R}^{n})=M_{0}^{1}(\mathbb{R}^{n})$ is the smallest
algebra containg the Schwartz functions and having properties (i) and (ii) above.

\section{The refined Robertson-Schr\"odinger uncertainty principle}

To prove our main theorem, we need the following two preliminary results.

\begin{proposition}\label{PropositionERSUP1}
Let $\widehat{A} \overset{\mathrm{Weyl}}{\longleftrightarrow} a$ be a positive Weyl operator with symbol $a \in \mathcal{S}^{\prime} (\mathbb{R}^{2n})$, and let $W \rho$ be the Wigner function associated with the density matrix $\widehat{\rho}$. If $\widehat{A}\widehat{\rho}$ is trace-class, then we have
\begin{equation}
\frac{1}{(2 \pi \hbar)^n} \int_{\mathbb{R}^{2n}} a(z) W \rho (z) dz \ge \int_{\mathbb{R}^{2n}} a(z) (W \rho (z) \star_{\hbar} W \rho (z)) dz \ge 0 ,
\label{eqERSUP1}
\end{equation}
where $\star_{\hbar}$ denotes the Moyal product. Moreover, the first inequality bccomes an equality if and only if the state is pure.
\end{proposition}

\begin{proof}
A density matrix is a trace class operator and hence compact. Thus, it admits the following spectral decomposition \cite{Gosson}:
\begin{equation}
\widehat{\rho} = \sum_{\alpha} \lambda_{\alpha} \widehat{P}_{\alpha},
\label{eqERSUP2}
\end{equation}
where $(\lambda_{\alpha})_{\alpha}$ are the eigenvalues of $\widehat{\rho}$, with
\begin{equation}
\lambda_{\alpha}>0, \hspace{1 cm} \sum_{\alpha} \lambda_{\alpha}=1.
\label{eqERSUP3}
\end{equation}
Here $\widehat{P}_{\alpha}$ is the orthogonal projection onto the eigenspace associated with the eigenvalue $\lambda_{\alpha}$.
Since
\begin{equation}
\widehat{\rho}^2 = \sum_{\alpha} \lambda_{\alpha}^2 \widehat{P}_{\alpha},
\label{eqERSUP4}
\end{equation}
we have by linearity, the positivity of $\widehat{A} $, convergence in the trace norm and the fact that $0 < \lambda_{\alpha} \le 1$:
\begin{equation}
0 \le Tr(\widehat{A}\widehat{\rho}^2) = \sum_{\alpha} \lambda_{\alpha}^2 Tr(\widehat{A}\widehat{P}_{\alpha}) \le \sum_{\alpha} \lambda_{\alpha} Tr(\widehat{A}\widehat{P}_{\alpha}) =Tr(\widehat{A}\widehat{\rho}) .
\label{eqERSUP5}
\end{equation}
Finally, an equality holds if and only if $\lambda_{\alpha}=0$ or $\lambda_{\alpha}=1$ for all $\alpha$. This is possible for a normalized state if and only if the state is pure. From (\ref{eqReview15},\ref{eqReview27}), we then recover (\ref{eqERSUP1}).
\end{proof}

The following technical result will also be useful

\begin{proposition}\label{PropositionERSUP2}
Let $F \in M_1^2 (\mathbb{R}^{2n})$ and $a(z)= \eta \cdot(z-z_0)$ for fixed $\eta \in \mathbb{C}^{2n}$ and $z_0 \in \mathbb{R}^{2n}$. Then the following identity holds:
\begin{equation}
\begin{array}{c}
\frac{1}{2} \int_{\mathbb{R}^{2n}} \left(|a \star_{\hbar}F|^2 + |F \star_{\hbar}a|^2 \right) dz=\\
\\
= \int_{\mathbb{R}^{2n}} \left(|a (z) |^2 |F(z)|^2 + \frac{ |\eta \cdot z|^2}{4} | (\mathcal{F}_{\sigma} F)(z)|^2 \right) dz.
\end{array}
\label{eqERSUP5.1}
\end{equation}
\end{proposition}

\begin{proof}
We start by showing that, as a distribution, $a \star_{\hbar}F \in \mathcal{S}^{\prime} (\mathbb{R}^{2n})$ is given by:
\begin{equation}
(a \star_{\hbar}F)(z)= a(z) F(z) + \frac{i \hbar}{2} \eta \cdot J \nabla F (z),
\label{eqERSUP5.2}
\end{equation}
where
\begin{equation}
\nabla F= \left(\frac{\partial F}{\partial x_1}, \cdots, \frac{\partial F}{\partial x_n} , \frac{\partial F}{\partial p_1}, \cdots, \frac{\partial F}{\partial p_n}\right)
\label{eqERSUP5.3}
\end{equation}
is the distributional gradient of $F$.

Indeed, let $\phi \in \mathcal{S}(\mathbb{R}^{2n})$. We have~by the distributional property (\ref{eqReview27}):
\begin{equation}
\begin{array}{c}
<a \star_{\hbar}F, \phi> =  <F,\phi \star_{\hbar} a > =\\
\\
= \int_{\mathbb{R}^{2n}} F(z) (\phi \star_{\hbar} a) (z) dz = \int_{\mathbb{R}^{2n}} F(z) \left(\phi (z) a (z) - \frac{i \hbar}{2} \eta \cdot J \nabla \phi (z) \right)  dz =\\
\\
= <Fa + \frac{i \hbar}{2} \eta \cdot J \nabla F, \phi> .
\end{array}
\label{eqERSUP5.4}
\end{equation}
Hence, (\ref{eqERSUP5.2}) follows.

Since $M_1^2( \mathbb{R}^{2n}) \simeq H^1 (\mathbb{R}^{2n}) \cap \mathcal{F}H^1 (\mathbb{R}^{2n})$ it is clear from (\ref{eqERSUP5.2}), that $a \star_{\hbar}F \in L^2 (\mathbb{R}^{2n})$. Moreover, given that
\begin{equation}
F \star_{\hbar}a= a \star_{- \hbar}F,
\label{eqERSUP5.5}
\end{equation}
the same can be said about $F \star_{\hbar}a$. We conclude that the left-hand side of (\ref{eqERSUP5.1}) is well defined and finite.

From (\ref{eqERSUP5.2},\ref{eqERSUP5.5}), we have
\begin{equation}
\begin{array}{c}
\frac{1}{2} \int_{\mathbb{R}^{2n}} \left(|a \star_{\hbar}F|^2 + |F \star_{\hbar}a|^2 \right) dz=\\
\\
= \frac{1}{2}\int_{\mathbb{R}^{2n}} \left(\left| a(z) F(z) + \frac{i \hbar}{2} \eta \cdot J \nabla F (z) \right|^2 + \left| a(z) F(z) - \frac{i \hbar}{2} \eta \cdot J \nabla F (z) \right|^2\right) dz =\\
\\
= \int_{\mathbb{R}^{2n}} \left(| a(z)|^2 |F(z)|^2 + \frac{\hbar^2}{4} | \eta \cdot J \nabla F(z)|^2 \right) dz.
\end{array}
\label{eqERSUP5.6}
\end{equation}
Since $F \in H^1 (\mathbb{R}^{2n}) $, we can express the last term as
\begin{equation}
\int_{\mathbb{R}^{2n}} | \eta \cdot J \nabla F(z)|^2 dz =  \frac{1}{\hbar^2} \int_{\mathbb{R}^{2n}} | \eta \cdot z|^2 |(\mathcal{F}_{\sigma} F)(z)|^2 dz
\label{eqERSUP5.7}
\end{equation}
and we recover (\ref{eqERSUP5.1}).
\end{proof}

We are now in a position to prove the refined RSUP. This uncertainty principle synthesizes the Heinig-Smith inequality and the RSUP, but is stronger than both.

\begin{theorem}\label{TheoremERSUP2}
Let $W \rho \in \mathcal{W}(\mathbb{R}^{2n})$ be such that
\begin{equation}
W \rho \in \mathcal{A} (\mathbb{R}^{2n}) := \left\{F \in M_{1}^{2} (\mathbb{R}^{2n}): ~ F \mbox{ is real and }  \operatorname*{Cov} (F) \mbox{ is finite } \right\}.
\label{eqERSUP6}
\end{equation}
Then the following matrix inequalities hold in $\mathbb{C}^{2n}$:
\begin{equation}
 \operatorname*{Cov}(W \rho) + \frac{i \hbar}{2}J \ge \mathcal{P} \left[ W \rho \right] \left(\operatorname*{Cov}(|\widetilde{W \rho}|^2) + \frac{1}{4} \operatorname*{Cov} (|\mathcal{F}_{\sigma}(\widetilde{W \rho})|^2) + \frac{i \hbar}{2}J  \right)  \ge 0 .
\label{eqERSUP7}
\end{equation}
The first inequality becomes a matrix identity if and only if the state is pure.

We remark that if a real function $F$ belongs to $M_1^2 (\mathbb{R}^{2n}) \cap M_2^1 (\mathbb{R}^{2n})$, then automatically $F \in \mathcal{A} (\mathbb{R}^{2n}) $.
\end{theorem}

\begin{proof}
We start by remarking that if $W \rho \in \mathcal{A} (\mathbb{R}^{2n})$, then all the covariance matrices appearing in (\ref{eqERSUP7}) are finite.

Define the operators
\begin{equation}
\widehat{Y}_j= \widehat{Z}_j - <\widehat{Z}_j> \widehat{I},
\label{eqERSUP8}
\end{equation}
for $j=1, \cdots, 2n$, and where
\begin{equation}
<\widehat{Z}_j>=Tr(\widehat{Z}_j \widehat{\rho}).
\label{eqERSUP9}
\end{equation}
Let also $\eta =(\eta_1, \cdots, \eta_{2n}) \in \mathbb{C}^{2n}$ and define
\begin{equation}
\widehat{A}:= (\eta \cdot \widehat{Y})^{\ast} (\eta \cdot \widehat{Y})= \sum_{j,k=1}^{2n} \overline{\eta_j} \eta_k \widehat{Y}_j\widehat{Y}_k,
\label{eqERSUP10}
\end{equation}
where $\widehat{B}^{\ast}$ denotes the adjoint of the operator $\widehat{B}$. Clearly, $\widehat{A}$ is a positive Weyl operator with symbol:
\begin{equation}
a(z)=\sum_{j,k=1}^{2n} \overline{\eta_j} \eta_k y_j\star_{\hbar} y_k =\sum_{j,k=1}^{2n} \overline{\eta_j} \eta_k(y_j y_k + \frac{i \hbar}{2} J_{jk} )= | \eta \cdot y|^2 + \frac{i \hbar}{2} \sigma (\eta, \overline{\eta}),
\label{eqERSUP11}
\end{equation}
where $y_j=z_j - <\widehat{Z}_j>$ is the symbol of $\widehat{Y}_j$.

Since $W \rho \in \mathcal{A} (\mathbb{R}^{2n})$, we have that $\widehat{A}\widehat{\rho}$ is trace class, or equivalently, that
\begin{equation}
\int_{\mathbb{R}^{2n}} a(z) W \rho (z) dz
\label{eqERSUP12}
\end{equation}
exists and is finite. We conclude that (\ref{eqERSUP1}) holds.

Next we evaluate the integrals in (\ref{eqERSUP1}). We start with
\begin{equation}
\begin{array}{c}
\int_{\mathbb{R}^{2n}} a(z) W \rho (z) dz = \sum_{j,k=1}^{2n} \overline{\eta_j} \eta_k \int_{\mathbb{R}^{2n}}y_jy_k W \rho (z) dz +   \frac{i \hbar}{2} \sigma (\eta, \overline{\eta}) = \\
\\
= \overline{\eta} \cdot  \operatorname*{Cov} (W \rho) \eta + \frac{i \hbar}{2} \sigma (\eta, \overline{\eta})= \overline{\eta} \cdot \left(\operatorname*{Cov} (W \rho)+ \frac{i \hbar}{2} J \right) \eta .
\end{array}
\label{eqERSUP12}
\end{equation}
Next, we have
\begin{equation}
\begin{array}{c}
(2 \pi \hbar)^n \int_{\mathbb{R}^{2n}} a(z) (W \rho (z) \star_{\hbar} W \rho (z)) dz = Tr(\widehat{A} \widehat{\rho}^2)=\\
\\
=Tr((\eta \cdot \widehat{Y})^{\ast} (\eta \cdot \widehat{Y}) \widehat{\rho}^2) =\\
\\
=\frac{1}{2}Tr\left( \left((\eta \cdot \widehat{Y})^{\ast} (\eta \cdot \widehat{Y})+(\eta \cdot \widehat{Y})(\eta \cdot \widehat{Y})^{\ast} \right) \widehat{\rho}^2\right)+\frac{1}{2}Tr\left(\left[(\eta \cdot \widehat{Y})^{\ast}, (\eta \cdot \widehat{Y})\right] \widehat{\rho}^2\right) =\\
\\
=\frac{1}{2}Tr\left[ \left((\eta \cdot \widehat{Y}) \widehat{\rho} \right) \left(\widehat{\rho}(\eta \cdot \widehat{Y})^{\ast} \right) \right]+\frac{1}{2}Tr\left[ \left((\eta \cdot \widehat{Y})^{\ast} \widehat{\rho} \right) \left(\widehat{\rho}(\eta \cdot \widehat{Y}) \right) \right] + \\
\\
+\frac{i \hbar}{2} \sigma( \eta , \overline{\eta}) Tr (\widehat{\rho}^2)=\\
\\
=\frac{(2 \pi \hbar)^n}{2} \int_{\mathbb{R}^{2n}} \left[ \left( (\eta \cdot y) \star_{\hbar} W \rho \right) \left(W \rho \star_{\hbar} (\overline{\eta \cdot y})  \right) + \right.\\
\\
\left. + \left( (\overline{\eta \cdot y}) \star_{\hbar} W \rho \right) \left(W \rho \star_{\hbar} (\eta \cdot y)  \right) \right]dz
+ \frac{i \hbar}{2}(2 \pi \hbar)^n \sigma( \eta , \overline{\eta}) ||| W \rho|||^2=\\
\\
=\frac{(2 \pi \hbar)^n}{2}  \int_{\mathbb{R}^{2n}}  \left( |(\eta \cdot y) \star_{\hbar} W \rho|^2 +| W \rho \star_{\hbar} (\eta \cdot y)|^2  \right) dz +\\
\\
+ \frac{i \hbar}{2}(2 \pi \hbar)^n \sigma( \eta , \overline{\eta}) ||| W \rho|||^2,
\end{array}
\label{eqERSUP13}
\end{equation}
where we used the cyclicity of the trace and (\ref{eqReview27}).

From Proposition \ref{PropositionERSUP2}, it follows that
\begin{equation}
\begin{array}{c}
(2 \pi \hbar)^n \int_{\mathbb{R}^{2n}} a(z) (W \rho (z) \star_{\hbar} W \rho (z)) dz = \\
\\
= (2 \pi \hbar)^n \int_{\mathbb{R}^{2n}} \left(|a(z)|^2 |W \rho (z)|^2 +\frac{|\eta \cdot z|^2}{4} |(\mathcal{F}_{\sigma} W \rho)(z) |^2\right) dz + \frac{i \hbar}{2} \sigma (\eta , \overline{\eta}) \mathcal{P}\left[W \rho \right].
\end{array}
\label{eqERSUP14}
\end{equation}
Now let us consider the two terms in the integral in previous expression. We have (recall that $<\widehat{Z}>$ is the expectation value for $W \rho$ and not $|W \rho|^2$):
\begin{equation}
\begin{array}{c}
 \int_{\mathbb{R}^{2n}}  |a(z)|^2 |W \rho (z)|^2 dz = \overline{\eta} \cdot \left(\int_{\mathbb{R}^{2n}}  (z-<\widehat{Z}>) (z-<\widehat{Z}>)^T |W \rho (z)|^2 dz \right) \eta \ge \\
 \\
 \ge min_{\zeta \in \mathbb{R}^{2n}} \left\{\overline{\eta} \cdot \left(\int_{\mathbb{R}^{2n}}  (z-\zeta) (z- \zeta)^T |W \rho (z)|^2 dz \right) \eta \right\} =\\
 \\
 = ||| W \rho|||^2 \overline{\eta} \cdot  \operatorname*{Cov} (|\widetilde{W \rho}|^2) \eta .
\end{array}
\label{eqERSUP15}
\end{equation}
Next, we remark that
\begin{equation}
\int_{\mathbb{R}^{2n}} |\eta \cdot z|^2 |(\mathcal{F}_{\sigma} W \rho)(z) |^2 dz = |||W \rho|||^2 \overline{\eta} \cdot  \operatorname*{Cov}\left( |\mathcal{F}_{\sigma} \widetilde{W \rho} |^2\right) \eta,
\label{eqERSUP16}
\end{equation}
where we used the fact that
\begin{equation}
\int_{\mathbb{R}^{2n}} z_j|\mathcal{F}_{\sigma} W \rho (z)|^2  dz=0,
\label{eqERSUP17}
\end{equation}
for $j=1, \cdots, 2n$, and that, by Placherel's Theorem, $|||\mathcal{F}_{\sigma} W \rho|||=|||W \rho|||$. Altogether, from (\ref{eqERSUP14})-(\ref{eqERSUP16}), we obtain
\begin{equation}
\begin{array}{c}
(2 \pi \hbar)^n \int_{\mathbb{R}^{2n}} a(z) \left( W \rho \star_{\hbar} W \rho \right) (z) dz \ge \\
\\
\mathcal{P} \left[W \rho \right] \overline{\eta} \cdot \left(  \operatorname*{Cov} (|\widetilde{W \rho}|^2) + \frac{1}{4} \operatorname*{Cov} \left( |\mathcal{F}_{\sigma} \widetilde{W \rho} |^2\right) + \frac{i \hbar}{2} J \right) \eta.
 \end{array}
\label{eqERSUP18}
\end{equation}
The first inequality in (\ref{eqERSUP7}) then follows from (\ref{eqERSUP1},\ref{eqERSUP12},\ref{eqERSUP18}).

To show the second inequality in (\ref{eqERSUP7}), we observe that, from our previous calculations (\ref{eqERSUP14}, \ref{eqERSUP16}):
\begin{equation}
\begin{array}{c}
|||W \rho|||^2 \overline{\eta} \cdot \left[  \operatorname*{Cov}(|\widetilde{W \rho}|^2) + \frac{1}{4} \operatorname*{Cov} (|\mathcal{F}_{\sigma} \widetilde{W \rho}|^2) + \frac{i \hbar}{2} J\right] \eta =\\
\\
= \mbox{min}_{\zeta \in\mathbb{R}^{2n}}  \overline{\eta} \cdot \left(  \int_{\mathbb{R}^{2n}} (z - \zeta)(z - \zeta)^T | W \rho(z)|^2 dz \right) \eta + \\
\\
+  |||W \rho|||^2 \overline{\eta} \cdot \left[ \frac{1}{4}  \operatorname*{Cov}(| \mathcal{F}_{\sigma} \widetilde{W \rho}|^2) + \frac{i \hbar}{2} J \right] \eta =\\
\\
= \mbox{min}_{\zeta \in\mathbb{R}^{2n}} \overline{\eta} \cdot \left( \int_{\mathbb{R}^{2n}} (z - \zeta) \star_{\hbar} (z- \zeta)^T (W \rho (z) \star_{\hbar} W \rho (z)) dz \right) \eta =\\
\\
= \mbox{min}_{\zeta \in\mathbb{R}^{2n}} \int_{\mathbb{R}^{2n}} b_{\zeta} (z) (W \rho (z) \star_{\hbar} W \rho (z)) dz =\\
\\
= \mbox{min}_{\zeta \in\mathbb{R}^{2n}}  \frac{1}{(2 \pi \hbar)^n} Tr (\widehat{B}_{\zeta} \widehat{\rho}^2),
\end{array}
\label{eqERSUP18.1}
\end{equation}
where $\widehat{B}_{\zeta}$ is the Weyl operator
\begin{equation}
\widehat{B}_{\zeta} = \left((\eta \cdot(\widehat{Z}- \zeta)\right)^{\ast} \left((\eta \cdot(\widehat{Z}- \zeta)\right),
\label{eqERSUP18.2}
\end{equation}
with symbol
\begin{equation}
 b_{\zeta} (z)  =  \overline{\eta} \cdot  (z - \zeta) \star_{\hbar} (z- \zeta)^T \eta = | \eta \cdot (z- \zeta)|^2 + \sigma (\eta, \overline{\eta}).
\label{eqERSUP18.3}
\end{equation}
This is manifestly a positive operator, and so from (\ref{eqERSUP18.1}), it follows that
\begin{equation}
|||W \rho|||^2 \overline{\eta} \cdot \left[  \operatorname*{Cov} (|\widetilde{W \rho}|^2) + \frac{1}{4} \operatorname*{Cov} (|\mathcal{F}_{\sigma} \widetilde{W \rho}|^2) + \frac{i \hbar}{2} J \right] \eta \ge 0.
\label{eqERSUP18.4}
\end{equation}

We leave to the reader the simple proof that the first inequality in (\ref{eqERSUP7}) becomes an equality if and only if the state is pure.
\end{proof}

\vspace{0.3 cm}
\noindent
The following is a simple corollary of the previous theorem.

\begin{corollary}\label{Corollary1}
Let $W \rho \in \mathcal{W}(\mathbb{R}^{2n}) \cap \mathcal{A} (\mathbb{R}^{2n})$. Then the following inequalities hold:
\begin{equation}
\begin{array}{l}
 \operatorname*{Cov}(W \rho) \ge \mathcal{P} \left[W \rho \right] \left(\operatorname*{Cov} (|\widetilde{W \rho}|^2) + \frac{1}{4} \operatorname*{Cov} (| \mathcal{F}_{\sigma} \widetilde{W \rho}|^2) \right), \\
\\
\operatorname*{Cov}(W \rho) \ge \mathcal{P} \left[W \rho \right] \operatorname*{Cov} (|\widetilde{W \rho}|^2), \\
\\
\operatorname*{Cov}(W \rho) \ge  \frac{\mathcal{P} \left[W \rho \right] }{4} \operatorname*{Cov} (| \mathcal{F}_{\sigma} \widetilde{W \rho}|^2) .
\end{array}
\label{eqERSUP18.4.B}
\end{equation}
\end{corollary}

\begin{proof}
The first inequality is obtained from (\ref{eqERSUP7}) by a restriction to $\mathbb{R}^{2n}$. The remaining two inequalities follow from the observation that $A +B \ge A$ if $A$ and $B$ are real symmetric and positive matrices.
\end{proof}

Before we proceed, we make the following remarks.

\begin{remark}\label{RemarkSymplecticCapacities}
The RSUP has an interesting geometric interpretation; as shown in
\cite{physreps} the condition%
\[
\Sigma+\frac{i\hbar}{2}J\geq0
\]
is equivalent to the condition $c(\Omega)\geq\pi\hbar$ where $\Omega$ is the
covariance ellipsoid and $c$ any symplectic capacity on the standard
symplectic space $(\mathbb{R}^{2n},\sigma)$. This property relates the RSUP to
deep results in symplectic topology (Gromov's non-squeezing theorem
\cite{Gromov}). It would certainly interesting to extend this geometric
interpretation to the refinement of the RSUP  and the inequalities (\ref{eqERSUP18.4.B}) proposed in the present paper.
\end{remark}

\begin{remark}\label{Remarkhbar}
Let $A \in M (n; \mathbb{C})$ be some complex matrix. Then $A$ is positive if and only if $A^T$ is positive. From this observation and the fact that $J^T =-J$ it follows that a function $F$ satisfies the refined RSUP (\ref{eqERSUP7}) if and only if it satisfies the same inequalities with $\hbar$ replaced by $- \hbar$.
\end{remark}

Next we show that the refined RSUP is invariant under linear symplectic and anti-symplectic transformations.

\begin{theorem}\label{TheoremSymplecticCovariance}
Suppose that $F \in \mathcal{A} (\mathbb{R}^{2n}) $ satisfies the refined RSUP:
\begin{equation}
  \operatorname*{Cov} (F) + \frac{i \hbar}{2} J \ge \mathcal{P} \left[F \right] \left( \operatorname*{Cov} (|\widetilde{F}|^2) + \frac{1}{4} \operatorname*{Cov} \left( |\mathcal{F}_{\sigma} \widetilde{F}|^2\right) + \frac{i \hbar}{2} J \right)  \ge 0 .
\label{eqERSUP19}
\end{equation}
Then for every $S \in ASp(n)$, the function $F \circ S$ also satisfies (\ref{eqERSUP19}).
\end{theorem}

\begin{proof}
A simple calculation shows that
\begin{equation}
\left(\mathcal{F}_{\sigma} (F \circ S)\right)( \zeta)= (\mathcal{F}_{\sigma} F) (\epsilon S\zeta),
\label{eqERSUP20}
\end{equation}
where $\epsilon=1$ if $S$ is symplectic and $\epsilon=-1$ if $S$ is anti-symplectic. It is then a straightforward task to check that
\begin{equation}
 \operatorname*{Cov}(G \circ S)= S^{-1} \operatorname*{Cov}(G) (S^{-1})^T,
\label{eqERSUP21}
\end{equation}
for $G=F,|\widetilde{F}|^2$ and $| \mathcal{F}_{\sigma} \widetilde{F}|^2$. Using the fact that $SJS^T= \epsilon J$, we conclude that $F \circ S$ satisfies (\ref{eqERSUP19}) with $\hbar$ replaced by $\epsilon \hbar$. In view of Remark \ref{Remarkhbar} the result follows.
\end{proof}

\begin{theorem}\label{TheoremSaturation}
Let $F\in \mathcal{A} (\mathbb{R}^{2n})$ be such that (\ref{eqERSUP19}) holds. Then $F$ has minimal Robertson-Schr\"odinger uncertainty,
\begin{equation}
\lambda_{\sigma,1} ( \operatorname*{Cov}(F))= \cdots=\lambda_{\sigma,n} (\operatorname*{Cov}(F))=\frac{\hbar}{2},
\label{eqsaturation1}
\end{equation}
if and only if $F$ is proportional to a Gaussian pure state Wigner function:
\begin{equation}
F(z) = \frac{1}{(\pi \hbar)^n} \exp \left( -\frac{1}{2} (z-z_0) \cdot ( \operatorname*{Cov}(F))^{-1} (z-z_0) \right)
\label{eqsaturation2}
\end{equation}
with $z_0 \in \mathbb{R}^{2n}$ and $\frac{2}{\hbar}  \operatorname*{Cov}(F) \in Sp (n)$.
\end{theorem}

\begin{proof}
Since (\ref{eqERSUP19}) holds, we have in particular
\begin{equation}
  \operatorname*{Cov}(F)  + \frac{i \hbar}{2} J \ge 0.
\label{eqsaturation3}
\end{equation}
Let $(u_j)_j$ be the $n$ eigenvectors of $ \operatorname*{Cov}(F) J^{-1}$ associated with the eigenvalues $-i \lambda_{\sigma,j} (\operatorname*{Cov}(F)) = -\frac{i \hbar}{2}$:
\begin{equation}
\operatorname*{Cov}(F) J^{-1} u_j = - \frac{i\hbar}{2} u_j, \hspace{1 cm} j=1, \cdots, n.
\label{eqsaturation4}
\end{equation}
Then we have:
\begin{equation}
\overline{u_j} \cdot J\left(  \operatorname*{Cov}(F)  + \frac{i \hbar}{2} J \right)J^{-1} u_j=0,
 \label{eqsaturation5}
\end{equation}
for $j=1, \cdots, n$.

From (\ref{eqERSUP19}), we must also have:
\begin{equation}
\overline{u_j} \cdot J \left(  \operatorname*{Cov}(|\widetilde{F}|^2)  +\frac{1}{4} \operatorname*{Cov} (|\mathcal{F}_{\sigma} \widetilde{F}|^2) +  \frac{i \hbar}{2} J \right) J^{-1} u_j=0,
 \label{eqsaturation5}
\end{equation}
for $j=1, \cdots, n$.

By (\ref{eqERSUP19}), the matrix
\begin{equation}
A= \operatorname*{Cov}(|\widetilde{F}|^2)  +\frac{1}{4} \operatorname*{Cov} (|\mathcal{F}_{\sigma} \widetilde{F}|^2)
 \label{eqsaturation6}
\end{equation}
satisfies the RSUP. And so, from (\ref{eqsaturation5}), we conclude that its symplectic eigenvalues
are also all equal to $\frac{\hbar}{2}$ and that $(u_j)_j$ are the associated eigenvectors:
\begin{equation}
A J^{-1} u_j = - \frac{i\hbar}{2} u_j, \hspace{1 cm} j=1, \cdots, n.
\label{eqsaturation7}
\end{equation}
It follows that
\begin{equation}
\det (A)= \Pi_{j=1}^n \left(\lambda_{\sigma,j}(A) \right)^2 = \left( \frac{\hbar}{2}\right)^{2n}.
\label{eqsaturation7.1}
\end{equation}
Setting $X=\det ( \operatorname*{Cov}(|\widetilde{F}|^2))$, $Y=\det (\frac{1}{4} \operatorname*{Cov} (|\mathcal{F}_{\sigma} \widetilde{F}|^2))$, we have from (\ref{eqsaturation7.1}), Minkowski's Determinant Theorem, the Heinig-Smith inequality (\ref{eqReview32.2}) and the arithmetic-geometric mean inequality that
\begin{equation}
\begin{array}{c}
\left(\frac{\hbar}{2} \right)^{2n} = \det \left( \operatorname*{Cov}(|\widetilde{F}|^2)  +\frac{1}{4} \operatorname*{Cov} (|\mathcal{F}_{\sigma} \widetilde{F}|^2) \right) \ge \\
\\
\ge \left( X^{\frac{1}{2n}} + Y^{\frac{1}{2n}} \right)^{2n} \ge  \left(2 \sqrt{ X^{\frac{1}{2n}}  Y^{\frac{1}{2n}}} \right)^{2n}=\\
\\
=\sqrt{\left(\det ( \operatorname*{Cov}(|\widetilde{F}|^2)) \right) \left(\det ( \operatorname*{Cov} (|\mathcal{F}_{\sigma} \widetilde{F}|^2)) \right)} \ge \left(\frac{\hbar}{2} \right)^{2n}.
\end{array}
\label{eqsaturation8}
\end{equation}
Thus all the inequalities become equalities. In particular the Heinig-Smith inequality is saturated, and $F$ must be of the form (\ref{eqReview32.3}).

We have
\begin{equation}
 \operatorname*{Cov} (F)=\frac{1}{2 \pi} A^{-1},   \hspace{0.5 cm} <z_j>_F=  (A^{-1}b)_j,
\label{eqsaturation9}
\end{equation}
for $j=1, \cdots, 2n$. Since, by assumption, $F$ is a real function, we conclude that $b \in \mathbb{R}^{2n}$, $c \in \mathbb{R}$ and $A$ is real, symmetric and positive-definite. Altogether, we recover (\ref{eqsaturation2}). Finally, since $F$ is a Gaussian distribution which saturates the RSUP, then by Littlejohn's Theorem we must have $\frac{2}{\hbar}  \operatorname*{Cov}(F) \in Sp (n)$.
\end{proof}

To complete our analysis we consider two examples. The first one shows that a function may satisfy the RSUP but not the refined RSUP. In a certain sense Example \ref{ExampleReview9} already does that. But that is not really a good example since $ \operatorname*{Cov}(|\mathcal{F}_{\sigma} \widetilde{F}|^2)$ is not finite.

The second example shows that the refined RSUP is not a sufficient condition for a phase space function to be a Wigner distribution.

\begin{example}\label{ExampleFinal1}
Consider the following real and normalized function defined on $\mathbb{R}^2$:
\begin{equation}
F(z)= \frac{48}{\pi \hbar} \left( \frac{|z|^2}{\hbar} - \frac{1}{6} \right) e^{- \frac{4 |z|^2}{\hbar}}.
\label{eqExampleFinal1}
\end{equation}
By straightforward calculations, we have:
\begin{equation}
 \operatorname*{Cov} (F)= \frac{\hbar}{2}I, \hspace{1 cm} \operatorname*{Cov} (|\widetilde{F}|^2) = \frac{11 \hbar}{80}I,
\label{eqExampleFinal2}
\end{equation}
while
\begin{equation}
\mathcal{P}\left[F \right] = 10.
\label{eqExampleFinal3}
\end{equation}
We conclude that
\begin{equation}
\operatorname*{Cov} (F) + \frac{i \hbar}{2}J \ge 0,
\label{eqExampleFinal4}
\end{equation}
that is $F$ satisfies the RSUP. On the other hand:
\begin{equation}
\mathcal{P}\left[F \right] \operatorname*{Cov} (|\widetilde{F}|^2) > \operatorname*{Cov}(F),
\label{eqExampleFinal5}
\end{equation}
which violates the second inequality in (\ref{eqERSUP18.4.B}).

To obtain a similar example in higher dimensions, we just have to take tensor products of the function (\ref{eqExampleFinal1}).
\end{example}

\begin{example}\label{ExampleFinal2}
Next consider the function
\begin{equation}
F(z) = \frac{1}{2 \pi \hbar} \left( \frac{|z|}{\hbar}-1 \right) e^{- \frac{|z|^2}{2 \hbar}} .
\label{eqExampleFinal6}
\end{equation}
A simple calculation shows that $\mathcal{F}_{\sigma}(F)=- F$ and that
\begin{equation}
 \operatorname*{Cov}(F)= 3 \hbar I, \hspace{0.7 cm} \operatorname*{Cov}(|\widetilde{F}|^2) = \operatorname*{Cov}(|\mathcal{F}_{\sigma} \widetilde{F}|^2) = \frac{3 \hbar}{2}I, \hspace{0.7 cm} \mathcal{P}\left[F \right]= \frac{1}{2}.
\label{eqExampleFinal7}
\end{equation}
We conclude that $F$ satisfies the refined RSUP (\ref{eqERSUP19}).

However, this is not a Wigner function. To see this consider the ground state of the simple harmonic oscillator:
\begin{equation}
F_0(z)= \frac{1}{\pi \hbar} e^{-\frac{|z|^2}{\hbar}}.
\label{eqExampleFinal8}
\end{equation}
We have:
\begin{equation}
\int_{\mathbb{R}^2} F(z) F_0(z) dz = - \frac{\hbar}{9},
\label{eqExampleFinal9}
\end{equation}
which violates the positivity condition (iv) in Theorem \ref{TheoremReview1}.
\end{example}

\section{The Hirschman-Shannon inequality for Wigner functions}

In this section, we prove the entropic inequalities which appear as a by-product of the refined RSUP.

\begin{theorem}\label{Theorementropy1}
Let $W\rho$ be a Wigner function with purity $\mathcal{P}%
\left[W\rho\right]$ and finite covariance matrix $Cov(W\rho)$. Then $|\widetilde{W\rho
}|^{2}$ and $|\mathcal{F}_{\hbar}\widetilde{W\rho}|^{2}$ have finite
covariance matrices and entropies and the following inequalities hold:
\begin{equation}%
\begin{array}
[c]{c}%
\log\left[  (2\pi e)^{2n}\det\left(  Cov(W\rho)\right)  \right]  \geq\\
\\
\geq\log\left[  \left(  \pi e\mathcal{P}\left[W\rho\right]\right)  ^{2n}\sqrt{\det\left(
Cov(|\widetilde{W\rho}|^{2})\right)  \cdot\det\left(  Cov(|\mathcal{F}_{\hbar
}\widetilde{W\rho}|^{2})\right)  }\right]  \geq\\
\\
\geq2n\log\left(  \mathcal{P}\left[W\rho\right]\right)  +E\left(  |\widetilde{W\rho}%
|^{2}\right)  +E\left(  |\mathcal{F}_{\hbar}\widetilde{W\rho}|^{2}\right)
\geq\log\left(  \pi\hbar e\mathcal{P}\left[W\rho\right]\right)  ^{2n} .%
\end{array}
\label{eqIntroductionentropy26}%
\end{equation}
We have an equality throughout in (\ref{eqIntroductionentropy26}) if and only if
$W \rho =W \psi$ is a pure Gaussian of the form:
\begin{equation}
W\psi(z)=\frac{1}{(\pi\hbar)^{n}}e^{-\frac{1}{2}(z-z_{0})\cdot\left(
Cov(W\psi)\right)  ^{-1}(z-z_{0})}, \label{eqIntroduction28}%
\end{equation}
where $z_{0}\in\mathbb{R}^{2n}$ and
\begin{equation}
\frac{2}{\hbar}Cov(W\psi)\in Sp(n) \label{eqIntroductionentropy29}%
\end{equation}
is a $2n\times2n$ real symplectic matrix.
\end{theorem}

\begin{proof}
\vspace{0.3 cm}
From (\ref{eqentropy6}) with $n \to 2n$ and $f=\widetilde{W\rho}$, we
obtain:
\begin{equation}%
\begin{array}
[c]{c}%
\log\left(  \pi\hbar e\right)  ^{2n}\leq E\left(  |\widetilde{W\rho}%
|^{2}\right)  +E\left(  |\mathcal{F}_{\hbar}\widetilde{W\rho}|^{2}\right)
\leq\\
\\
\leq\log\left[  (2\pi e)^{2n}\sqrt{\det\left(  Cov(|\widetilde{W\rho}%
|^{2})\right)  \cdot\det\left(  Cov(|\mathcal{F}_{\hbar}\widetilde{W\rho}%
|^{2})\right)  }\right] .
\end{array}
\label{eqentropy21}%
\end{equation}

The first inequality in (\ref{eqERSUP18.4.B}) and Minkowski's determinant theorem yield
\begin{equation}%
\begin{array}
[c]{c}%
\det\left(  Cov(W \rho)\right)  \ge\left(  \mathcal{P} \left[W \rho\right] \right)  ^{2n}
\det\left[  Cov \left(  |\widetilde{W \rho}|^{2}\right)  + \frac{1}{4} J^{T}
Cov \left(  |\mathcal{F}_{\hbar} \widetilde{W \rho}|^{2}\right)  J\right]
\ge\\
\\
\ge\left(  \mathcal{P} \left[W \rho \right] \right)  ^{2n} \left[  \det^{\frac{1}{2n}}
\left(  Cov \left(  |\widetilde{W \rho}|^{2}\right)  \right)  + \frac{1}{4}
\det^{\frac{1}{2n}} \left(  Cov \left(  |\mathcal{F}_{\hbar} \widetilde{W
\rho}|^{2}\right)  \right)  \right]  ^{2n} .%
\end{array}
\label{eqProof3}%
\end{equation}
From the concavity of the logarithm and (\ref{eqentropy21}):
\begin{equation}%
\begin{array}
[c]{c}%
log \left(  \det\left(  Cov(W \rho)\right)  \right)  \ge2n \log\left(  2
\mathcal{P} \left[W \rho\right] \right)  +\\
\\
+2n \log\left[  \frac{1}{2} \det^{\frac{1}{2n}} \left(  Cov \left(
|\widetilde{W \rho}|^{2}\right)  \right)  + \frac{1}{8} \det^{\frac{1}{2n}}
\left(  Cov \left(  |\mathcal{F}_{\hbar} \widetilde{W \rho}|^{2}\right)
\right)  \right]  \ge\\
\\
\ge2n \log\left(  \mathcal{P} \left[W \rho\right] \right)  + \frac{1}{2} \log\left(  Cov
\left(  |\widetilde{W \rho}|^{2}\right)  \right)  + \frac{1}{2} \log\left(
Cov \left(  |\mathcal{F}_{\hbar} \widetilde{W \rho}|^{2}\right)  \right)
\ge\\
\\
\ge2n \log\left(  \mathcal{P} \left[W \rho\right] \right)  + E \left(  |\widetilde{W
\rho}|^{2}\right)  + E \left(  |\mathcal{F}_{\hbar} \widetilde{W \rho}%
|^{2}\right)  - \log(2 \pi e)^{2n},
\end{array}
\label{eqProof4}%
\end{equation}
and the result follows.

Finally, suppose we have an equality throughout (\ref{eqIntroductionentropy26}). By
Hirschman's Theorem, the last inequality becomes an equality if and only if $W
\rho$ is a generalized Gaussian. But since, $W \rho$ is a real normalized
function, it must be of the form:
\begin{equation}
W \rho(z) = \frac{1}{(2 \pi)^{n} \sqrt{\det A}} e^{- \frac{1}{2} (z-z_{0})
\cdot A^{-1} (z-z_{0})}, \label{eqProof5}%
\end{equation}
with $A$ a real, symmetric, positive-definite $2n \times2n $ matrix. By
standard Gaussian integral computations, we conclude that:
\begin{equation}
\mathcal{P} \left[W \rho\right]= \left(  \frac{\hbar}{2} \right)  ^{n} \frac{1}%
{\sqrt{\det A}}, \hspace{1 cm} Cov (W \rho)=A. \label{eqProof6}%
\end{equation}
Moreover,
\begin{equation}%
\begin{array}
[c]{l}%
\widetilde{W \rho} (z)= \frac{1}{\pi^{n/2} \sqrt[4]{\det A}} e^{- \frac{1}{2}
(z-z_{0}) \cdot A^{-1} (z-z_{0})},\\
\\
\left(  \mathcal{F}_{\hbar} \widetilde{W \rho} \right)  (\zeta)=
\frac{\sqrt[4]{\det A}}{\pi^{n/2} \hbar^{n}} e^{- \frac{1}{2 \hbar^{2}}
\zeta\cdot A \zeta- \frac{i}{\hbar} \zeta\cdot z_{0}}.
\end{array}
\label{eqProof7}%
\end{equation}
From which we conclude that
\begin{equation}
Cov \left(  |\widetilde{W \rho}|^{2}\right)  = \frac{1}{2} A , \hspace{1 cm}
Cov \left(  |\mathcal{F}_{\hbar} \widetilde{W \rho}|^{2}\right)  = \frac{1}{2
\hbar^{2}} A^{-1}. \label{eqProof8}%
\end{equation}

If we have an equality throughout (\ref{eqIntroductionentropy26}), then we also have
an equality in (\ref{eqProof3}). By Minkowski's determinant theorem that can
happen if and only if, there exists a constant $\alpha\ge0$, such that
\begin{equation}
Cov \left(  |\widetilde{W \rho}|^{2}\right)  = \alpha J^{T} Cov \left(
|\mathcal{F}_{\hbar} \widetilde{W \rho}|^{2}\right)  J \label{eqProof9}%
\end{equation}
Plugging (\ref{eqProof8}) into (\ref{eqProof9}) yields:
\begin{equation}
A= \frac{\alpha}{\hbar^{2}} J^{T} A^{-1} J \Leftrightarrow A J A =
\frac{\alpha}{\hbar^{2}} J .
\label{eqProof10}%
\end{equation}
In other words: $A$ is proportional to a symplectic matrix.

Equating the first and the last term in (\ref{eqIntroductionentropy26}), we obtain:
\begin{equation}
\det\left(  Cov (W \rho) \right)  = \left(  \frac{\hbar}{2} \right)  ^{2n}
\mathcal{P}^{2n} (W \rho) .
\label{eqProof11}%
\end{equation}
From (\ref{eqProof6}) and (\ref{eqProof11}), we conclude that:
\begin{equation}
\det A = \left(  \frac{\hbar}{2} \right)  ^{2n}, \hspace{1 cm} \mathcal{P} \left[W
\rho \right]=1. \label{eqProof12}%
\end{equation}
which proves the result.
\end{proof}

Another consequence of the refined RSUP is the following corollary for pure states.

\begin{corollary}
\label{Corollary3} Suppose that the Wigner function $W \psi$ has a finite
covariance matrix. Then $|\widetilde{W \psi}|^{2}$ has a finite covariance
matrix and a finite entropy and we have:
\begin{equation}%
\begin{array}
[c]{c}%
\log\left[  (2\pi e)^{n} \sqrt{\det\left(  Cov (W \psi) \right)  }\right]
\ge\\
\\
\ge\log\left[  \left(  2\pi e \right)  ^{n} \sqrt{\det\left(  Cov
(|\widetilde{W \psi}|^{2}) \right)  } \right]  \ge\\
\\
\ge E \left(  |\widetilde{W \psi}|^{2}\right)  \ge\log\left(  \frac{ \pi\hbar
e}{2} \right)  ^{2n} .%
\end{array}
\label{eqIntroductionentropy30}%
\end{equation}
\end{corollary}

\begin{proof}
The last inequality in (\ref{eqIntroductionentropy30}) is a well known result by E.
Lieb \cite{Lieb}. The penultimate inequality is just Shannon's inequality
(\ref{eqentropy3}). In remains to prove the first inequality. But again from
the first inequality in (\ref{eqERSUP18.4.B}), we conclude that
\begin{equation}
\det\left(  Cov (W \psi) \right)  \ge\det\left(  Cov (| \widetilde{W \psi
}|^{2} ) \right)  ,
\end{equation}
and the result follows.
\end{proof}

\begin{remark}
\label{Remark2} Notice that the previous results are mainly interesting if the
state $W\rho$ does not depart appreciably from a pure state, that is if
$\mathcal{P}\left[W\rho \right]\approx1$. This is of course true if we have exactly a pure
state as in (\ref{eqIntroductionentropy30}). If $W\rho$
is highly mixed $\mathcal{P}\left[W\rho \right]\approx0$, then $\log\left(  \mathcal{P}%
\left[W\rho \right]\right)  \rightarrow-\infty$, and inequality (\ref{eqIntroductionentropy26})
becomes trivially true.
\end{remark}

\begin{remark}
Before we proceed let us make a brief comment on the choice of Fourier transform in the various inequalities.
In the refined RSUP we chose the symplectic Fourier transform in order to have a simpler expression. Otherwise,
we would have to make the replacement
\begin{equation}
Cov \left(|\mathcal{F}_{\sigma} (\widetilde{W \rho})|^2 \right) =J^T  Cov \left(|\mathcal{F}_{\hbar} (\widetilde{W \rho})|^2 \right) J .
\label{eqentropy26.1}%
\end{equation}
Because of this identity, the determinants of the two covariance matrices coincide. Likewise, we can easily show that $E \left( |\mathcal{F}_{\sigma} (\widetilde{W \rho})|^2\right)= E \left( |\mathcal{F}_{\hbar} (\widetilde{W \rho})|^2\right)  $. Consequently, (\ref{eqentropy26}) holds whether we use $|\mathcal{F}_{\sigma} (\widetilde{W \rho})|^2$ or $|\mathcal{F}_{\hbar} (\widetilde{W \rho})|^2$. We picked $|\mathcal{F}_{\hbar} (\widetilde{W \rho})|^2$ because we can then compare it directly with the Hirschman inequality. But this is really just a question of taste.
\end{remark}

\section{Outlook}

The Wigner quasi-distribution plays a central role in both time-frequency
analysis and quantum mechanics (from which it originates). One should however
be aware that it is not the only possible choice. Any element of the so-called
Cohen class \cite{Gro} having the correct marginals is a priory an
equally good choice in entropic questions of the type considered in this paper
(even if the Wigner quasi-distribution is well-adapted when symplectic
symmetries are present). It would for instance be interesting to generalize
our results to a particular element of the Cohen class, namely the
Born--Jordan distribution \cite{Springer} which is closely related to the
eponymous quantization procedure, and which has certain advantages compared to
those of the Wigner quasi-distribution (in particular it damps certain
unwanted interference effects \cite{cogoni}). We hope to come back to this
case in a near future.

\section*{Acknowledgements}

The work of N.C. Dias and J.N. Prata is supported by the COST Action 1405 and
by the Portuguese Science Foundation (FCT) grant PTDC/MAT-CAL/4334/2014. M. de
Gosson has been funded by the grant P27773 of the Austrian Research agency FWF.

\pagebreak

**********************************************************************************************************************************************************************************************************

\textbf{Author's addresses:}

\begin{itemize}
\item \textbf{Nuno Costa Dias and Jo\~ao Nuno Prata: }Escola Superior N\'autica Infante D. Henrique. Av.
Eng. Bonneville Franco, 2770-058 Pa\c{c}o d'Arcos, Portugal and Grupo de F\'{\i}sica
Matem\'{a}tica, Departamento de Matem\'atica, Faculdade de Cí\^encias, Universidade de Lisboa, Campo Grande, Edif\'{\i}cio C6, 1749-016 Lisboa, Portugal

\item \textbf{Maurice A. de Gosson:} Universit\"{a}t Wien, Fakult\"{a}t
f\"{u}r Mathematik--NuHAG, Nordbergstrasse 15, 1090 Vienna, Austria
\end{itemize}

**********************************************************************************************************************************************************************************************************

\end{document}